\documentclass[journal]{IEEEtran}
\usepackage{amsmath,amssymb,amsfonts,bm}
\usepackage{amsthm}
\usepackage{booktabs}
\usepackage{graphicx}
\usepackage{cite}
\usepackage[colorlinks=false,hidelinks]{hyperref}
\usepackage{orcidlink}

\newcommand{\im}{\jmath}                        
\newcommand{\Imag}{\operatorname{Im}}
\newcommand{\Real}{\operatorname{Re}}
\newcommand{\sinc}{\operatorname{sinc}}
\newcommand{\diag}{\operatorname{diag}}
\newcommand{\dd}{\mathrm{d}}
\newcommand{\bs}{\mathbf{s}}
\newcommand{\br}{\mathbf{r}}
\newcommand{\bu}{\bar{\mathbf{u}}}
\newcommand{\bk}{\bar{\mathbf{k}}}
\newcommand{\bp}{\mathbf{p}}
\newcommand{\bw}{\mathbf{w}}
\newcommand{\bh}{\mathbf{h}}
\newcommand{\bG}{\mathbf{G}}
\newcommand{\bI}{\mathbf{I}}
\newcommand{\bC}{\mathbf{C}}
\newcommand{\bB}{\mathbf{B}}

\newcommand{\bZ}{\mathbf{Z}}
\newcommand{\bj}{\mathbf{j}}
\newcommand{\bq}{\mathbf{q}}
\newcommand{\surf}{\mathcal{S}}
\newcommand{\Utt}{\mathcal{U}}
\newcommand{\kerZ}{\boldsymbol{\mathcal{Z}}}
\newcommand{\kerR}{\boldsymbol{\mathcal{R}}}
\newcommand{\kerX}{\boldsymbol{\mathcal{X}}}

\newtheorem{theorem}{Theorem}
\newtheorem{proposition}{Proposition}
\newtheorem{lemma}{Lemma}
\newtheorem{corollary}{Corollary}
\newtheorem{remark}{Remark}

\begin{document}

\title{Mutual-Coupling-Aware Movable and Fluid Antennas on Holographic Surfaces: A Wavenumber-Domain Circuit-Field Unification}

\author{Giovanni Iacovelli$^{\orcidlink{0000-0002-3551-4584}}$,~\IEEEmembership{Member,~IEEE}, Chandan~Kumar~Sheemar$^{\orcidlink{0000-0003-1676-5983}}$,~\IEEEmembership{Member,~IEEE}, and \\Symeon Chatzinotas$^{\orcidlink{0000-0001-5122-0001}}$,~\IEEEmembership{Fellow,~IEEE}%
\thanks{The authors are with the Signal Processing and Communications (SIGCOM) Research Group at Interdisciplinary Centre for Security, Reliability and Trust (SnT), University of Luxembourg, 1855 Luxembourg City, Luxembourg.}}

\maketitle

\begin{abstract}
Movable and fluid antenna systems turn antenna position into a design variable. At sub-wavelength spacings, however, their behavior is governed by mutual coupling, modeled today by two disjoint traditions: circuit-theoretic impedance matrices with element-level constants, and field-theoretic kernels with norm-type power constraints. This paper unifies the two. Starting from the impedance kernel of a holographic surface, a Poynting-anchored balance identifies its resistive part with ohmic plus radiated power and its reactive part with stored-energy imbalance, and a circuit-field equivalence shows that the multiport impedance matrix is the kernel sampled at the port separations, in a single closed spherical-Hankel form. In the wavenumber domain the resistive kernel asymptotically diagonalizes in the aperture size: visible modes radiate at closed-form prices, evanescent modes only dissipate, and a flexible port becomes a constant-modulus codeword whose coupling is the pullback of the spectral weight. Coupling-aware multi-user sum-rate maximization over precoders and port positions is then formulated under physical power and voltage constraints and solved by weighted-MMSE and projected-gradient steps with closed-form gradients. A modal relaxation upper-bounds every port configuration and seeds the search by FFT-based codeword projection. A half-wavelength corollary and a superdirectivity margin quantify when coupling hurts, and when it helps.
\end{abstract}


\section{Introduction}
\IEEEPARstart{F}{lexible} antenna positioning elevates where a radiator sits to the rank of design variable, as consequential as what it transmits, and this paper treats its two architectures as one. Movable antennas (MAs) \cite{ZhuMA2024Mag,ZhuMA2024TWC,MaMIMOMA2024} continuously reposition discrete radiators at sub-wavelength scale, while fluid antenna systems (FASs) \cite{Wong2021FAS,New2024Tutorial} let a single radiator, liquid-based, pixel-based, or mechanically actuated, occupy any of a large number of ports within a prescribed surface, and switching to the port with the most favorable fading realization yields diversity and multiplexing gains that rival those of much larger conventional arrays \cite{Wong2021FAS,New2024MIMOFAS,Chai2022Port,Wong2022FAMA}. The appeal of the concept lies precisely in its compactness: the ports are packed far below the half-wavelength spacing of classical arrays, so that a hand-held device can skim through hundreds of channel realizations within a fraction of a wavelength. Throughout, MA and FAS are one mathematical object, an antenna whose admissible positions form a continuum (MA) or a finite grid (FAS). We write ``port'' and ``position'' generically, and Section~\ref{sec:fasequiv} makes the unification exact.

This compactness, however, is also where the standard positional abstraction becomes fragile. Communication-theoretic treatments almost invariably describe the ports as independent samples of the channel, correlated through the fading statistics but electromagnetically inert with respect to one another. The abstraction is defensible when a single port is active, the coupling of a radiator with itself is absorbed into its radiation resistance, but it collapses as soon as several ports radiate simultaneously, as in MIMO-FAS \cite{New2024MIMOFAS}, fluid antenna arrays \cite{ZhangEMFAA2026}, or pixel-based implementations where the passive pixels load the active ones even when switched off \cite{New2024Tutorial}. At sub-wavelength spacing the currents on distinct ports interact through the near field, the radiated power ceases to be the sum of per-port powers, and both the achievable beam patterns and the very meaning of the transmit power constraint change. In short: the regime in which FAS is most attractive is exactly the regime in which its standard model is least valid.

\subsection{Two modeling cultures, and a domain in between}
The literature that does take coupling seriously splits into two cultures. The first is circuit-first: antennas are ports of a multiport network \cite{Ivrlac2010}, coupling is encoded in mutual impedances, and the communication model is derived from Kirchhoff-consistent voltage-current relations. This is the route taken by the mutual-coupling appendix of the MIMO-FAS framework \cite{New2024MIMOFAS}, where fixed pre-computed impedance matrices multiply the channel from both sides but remain outside the optimization. By the recent movable-antenna (MA) work of Liao et al. \cite{Liao2026MC}, which optimizes antenna positions and transmit covariance through analytically intractable impedance matrices via trust-region machinery and Sylvester equations, and demonstrates that coupling, far from being an impairment, is a source of capacity gain through superdirectivity. By the electromagnetic-aware fluid antenna array framework of Zhang et al. \cite{ZhangEMFAA2026}, which integrates position-dependent multiport impedances into a current-domain description and jointly optimizes currents and positions for superdirective beamforming and multi-user sum rate, and by load-tuning architectures that manipulate coupling with parasitic elements \cite{Faddoul2025Loads}. Metasurface-based FAS embodiments admit the same treatment through admittance matrices \cite{Ramirez2025Meta,Ramirez2026Beamspace}, and dynamic metasurface antennas have been shown, both theoretically and experimentally, to benefit from strong coupling \cite{delHougne2025Benefits}.

The second culture is field-first, and it is native to holographic MIMO (HMIMO) and continuous-aperture arrays (CAPAs): the transmitter is a current density on a surface, propagation is the dyadic Green function, and coupling is not a matrix but an operator. Pizzo and Lozano \cite{PizzoLozano2025} showed that, for holographic apertures, correlation and coupling are two spatial filters acting in cascade, the former a convolution, the latter a deconvolution, and that antenna patterns can be chosen so that coupling reinforces or counters correlation as the SNR dictates. Coupling is likewise known to reshape the degrees of freedom and radiation efficiency of space-constrained apertures \cite{Yuan2023DoF}. Most recently, Wang et al. \cite{WangCAPA2025} developed a physical CAPA coupling model that retains both polarization and ohmic surface dissipation, characterized the wavenumber-domain coupling kernel, showed that the coupled beampattern is the uncoupled one filtered by that kernel, quantified how polarization displaces the classical half-wavelength coupling nulls, and extended the continuous model to discrete arrays through element current profiles. Coupling-aware multi-user CAPA beamforming has been formulated along the same lines \cite{WangCAPA2026MU}. The wavenumber (spatial-Fourier) representation itself is by now a mature signal-processing substrate for HMIMO, with plane-wave channel expansions \cite{Pizzo2022Fourier}, wavenumber-division multiplexing \cite{Sanguinetti2023WDM}, and a dedicated survey \cite{ZhangDai2026Survey}. It also underpins operator-theoretic mode designs for multi-user holographic transmission \cite{IacovelliHMISO} as well as wavenumber-domain sensing constructions \cite{IacovelliVirtual}.

What makes this second culture directly relevant to FAS is a structural observation made in the channel-estimation literature: a fluid antenna over a surface is a holographic aperture with antenna selection, the port channels being samples of a band-limited electromagnetic field \cite{New2025Oversampling}. The observation has so far been exploited for estimation and reconstruction, where it explains why oversampling below half a wavelength is essential, but not for coupling-aware design. Independently, Fourier-type syntheses have connected beampatterns to port selection \cite{Xu2026Flexible} and antenna positioning \cite{MA2026IFFT}, yet without any electromagnetic coupling in the loop. The flagship wavenumber-domain survey \cite{ZhangDai2026Survey} contains, tellingly, neither fluid antennas nor mutual coupling in its treated basis.

\subsection{The gap}
Table~\ref{tab:sota} organizes the state of the art along the axes that matter here: architecture, coupling model, working domain, and optimization scope. Every coupling-aware FAS/MA optimization to date \cite{Liao2026MC,ZhangEMFAA2026,Faddoul2025Loads} operates in the spatial or current domain with circuit-theoretic matrices, never in the wavenumber representation where the coupling operator becomes asymptotically diagonal. Every wavenumber-domain coupling treatment \cite{PizzoLozano2025,WangCAPA2025,WangCAPA2026MU} concerns fixed apertures, and the FAS-HMIMO bridge \cite{New2025Oversampling} is coupling-free, the only multi-user instance \cite{ZhangEMFAA2026} optimizing continuous currents rather than port selection. The intersection cell, a wavenumber-domain, coupling-consistent formulation of flexible positioning and port selection for multi-user transmission, is empty, and this paper fills it. As a by-product, to the best of our knowledge no prior multi-user FAS/MA formulation imposes the physically correct coupled power and voltage constraints, under which the budget is a position-dependent quadratic form of the impedance kernel rather than a sum of per-port powers.

\begin{table}[!t]
\centering
\caption{State of the art on mutual coupling (MC) in flexible-position and holographic architectures.}
\label{tab:sota}
\scriptsize
\setlength{\tabcolsep}{3pt}
\begin{tabular}{@{}lllll@{}}
\toprule
Reference & Architecture & MC model & Domain & Optimization \\
\midrule
\cite{ZhuMA2024Mag,ZhuMA2024TWC,MaMIMOMA2024} & MA & none & spatial & model{+}pos.{+}cov. \\
\cite{New2024MIMOFAS} & MIMO-FAS & fixed-load $\bZ$ & spatial & select.{+}BF \\
\cite{Liao2026MC} & MA & circuit $\bZ$ & spatial & pos.{+}cov. \\
\cite{ZhangEMFAA2026} & FA array & multiport $\bZ$ & current & currents{+}pos. \\
\cite{Ramirez2025Meta,Ramirez2026Beamspace} & DMA-FAS & admittance & beamspace & codebook \\
\cite{Faddoul2025Loads,delHougne2025Benefits} & loaded/DMA & circuit & spatial & loads{+}atoms \\
\cite{PizzoLozano2025} & HMIMO & circuit$\to$Fourier & wavenumber & patterns \\
\cite{Yuan2023DoF} & HMIMO & corr.{+}efficiency & spatial & analysis \\
\cite{WangCAPA2025,WangCAPA2026MU} & CAPA & resistive kernel & wavenumber & beamforming \\
\cite{ZhangDai2026Survey} & HMIMO & none & wavenumber & survey \\
\cite{New2025Oversampling} & FAS & none & field (Nyq.) & estimation \\
\cite{Xu2026Flexible,MA2026IFFT} & FAS/MA & none & sp.\ Fourier & selection \\
\midrule
\textbf{This work} & \textbf{FAS/MA on surface} & \textbf{full kernel $\kerZ$} & \textbf{wavenumber} & \textbf{pos.{+}precoding} \\
\bottomrule
\end{tabular}
\end{table}

\subsection{Contributions}
Starting from electromagnetic theory, we build a coupling-consistent multi-user communication model on a holographic surface, embed the flexible-position antenna in it as a structured wavenumber-domain object, and optimize it. Specifically:
\begin{itemize}
\item \textbf{Full impedance-kernel model with coupled multi-user constraints.} The downlink of a surface transmitter serving $K$ dipole users is formulated with coupling entering where physics puts it, through the impedance kernel $\kerZ=Z_s\delta\,\bI_3-\im\kappa Z_0\bG$, whose Hermitian split $\kerZ=\kerR+\im\kerX$ separates the resistive part governing average power from the reactive part governing port voltages and stored energy, all in closed spherical-Bessel form (Section~\ref{sec:model}).
\item \textbf{Circuit-field equivalence, settled in both directions.} The multiport impedance matrix of any port array is the pullback of $\kerZ$ under the element current profiles (Theorem~\ref{thm:equiv}): the mutual entries and the radiation self-resistance admit universal point-source limits, the latter reproducing the Hertzian $Z_0\kappa^2\ell^2/(6\pi)$ exactly, while the self-reactance and ohmic self-resistance provably do not, which is why circuit models carry them as element-specific constants (e.g., $73.08+\im42.21\,\Omega$ for the half-wave dipole \cite{New2024MIMOFAS}). The two cultures of Table~\ref{tab:sota} therefore agree wherever both are defined.
\item \textbf{An exact wavenumber dichotomy and diagonalization.} The light circle $\|\bk\|=\kappa$ splits the kernel exactly, resistive part on the visible disc plus ohmic floor, reactive part on the evanescent region (Proposition~\ref{prop:split}). A Parseval identity shows that only the visible spectrum radiates, Fourier modes asymptotically diagonalize the resistive kernel with explicit weight (Lemma~\ref{lem:diag}), and sharply truncated modes have divergent reactive energy, the reason reactive bookkeeping lives at port level.
\item \textbf{Positioning-HMIMO equivalence.} A flexible-position antenna with $A$ active ports is an HMIMO transmitter confined to an $A$-fold superposition of constant-modulus wavenumber codewords: moving a port is a spectral modulation, selecting one a codeword selection, and the coupling between ports is the pullback of the kernel (Theorem~\ref{thm:fas}). The pullback reproduces the transcendental null equations of \cite{WangCAPA2025} and collapses to uncoupled MIMO at half wavelength up to the residual $3/(2\pi^2m^2)$, whose reactive counterpart is an order of magnitude larger, $\approx0.43$ of the radiation resistance (Corollary~\ref{cor:halflambda}).
\item \textbf{Coupling-aware sum-rate maximization.} Multi-user sum-rate maximization over positions and precoders is cast under the coupled power constraint, regularized by the ohmic ratio that caps superdirective gains, with optional voltage constraints from the reactive kernel, and solved by alternating weighted-MMSE precoding with projected-gradient position updates with every gradient closed form, unlike circuit-first formulations requiring numerical matrix differentiation \cite{Liao2026MC}. The equivalence is exploited algorithmically: the modal relaxation furnishes a holographic upper bound, and its FFT projection onto the codeword manifold initializes the search and certifies the price of $A$-port sparsity (Section~\ref{sec:opt}).
\end{itemize}

\subsection{Notation}
Boldface lowercase (uppercase) letters denote vectors (matrices), calligraphic letters sets, surfaces, and kernels, and $(\cdot)^{\mathsf T}$, $(\cdot)^{\mathsf H}$, $(\cdot)^{*}$ transpose, conjugate transpose, and conjugate. $\Real\{\cdot\}$ and $\Imag\{\cdot\}$ act entrywise, $\im$ is the imaginary unit, and the convention $e^{-\im2\pi ft}$ makes outgoing waves carry $e^{+\im\kappa R}$. Points are $\br=[\bar{\br}^{\mathsf T},z]^{\mathsf T}$ with planar part $\bar{\br}\in\mathbb{R}^2$, planar wavevectors are $\bk$ with $k_z=\sqrt{\kappa^2-\|\bk\|^2}$, $\Imag\{k_z\}\ge0$, and $\kappa=2\pi/\lambda$. $j_\ell$, $y_\ell$, $h_\ell=j_\ell+\im y_\ell$ are the spherical Bessel and Hankel functions, $\sinc(x)=j_0(x)$, $\delta_{nm}$ the Kronecker delta, $\bI_d$ the identity. For a kernel $\boldsymbol{\mathcal{K}}(\bs,\bs')$ on $\surf$ we write $\langle \bj_1,\boldsymbol{\mathcal{K}}\bj_2\rangle \triangleq \iint_{\surf\times\surf}\bj_1^{\mathsf H}(\bs)\boldsymbol{\mathcal{K}}(\bs,\bs')\bj_2(\bs')\,\dd\bs\,\dd\bs'$. Electromagnetic fields are set in boldface calligraphic type ($\boldsymbol{\mathcal{E}}$, $\boldsymbol{\mathcal{H}}$) to distinguish them from channel and precoding vectors.

\section{Electromagnetically Consistent Multi-User Communication Model}\label{sec:model}
We begin where the physics begins, with currents, fields, and power, and obtain a multi-user downlink model in which coupling is not appended as a corrective matrix but appears in the constraints through the full impedance kernel, ohmic, radiative, and reactive alike. Section~\ref{sec:wavenumber} then moves the model to the wavenumber domain, where the flexible-position antenna finds its natural description.

\subsection{Propagation and received signal}
Consider a quasi-static narrowband system in which a planar surface $\surf=[-S_x/2,S_x/2]\times[-S_y/2,S_y/2]\subset\{z=0\}$, of area $|\surf|=S_xS_y$, serves $K$ single-antenna users located at $\br_k=[\bar{\br}_k^{\mathsf T},z_k]^{\mathsf T}$, $z_k>0$, through an infinite homogeneous medium. The transmitter synthesizes a monochromatic surface current density $\bj(\bs)\in\mathbb{C}^3$ [A/m], tangential to $\surf$, radiating the electric field
\begin{equation}\label{eq:field}
\boldsymbol{\mathcal{E}}(\br)=\im\kappa Z_0\int_{\surf}\bG(\br,\bs)\,\bj(\bs)\,\dd\bs \quad \text{[V/m]},
\end{equation}
where $Z_0$ is the characteristic impedance of vacuum and
\begin{equation}\label{eq:green}
\bG(\br,\bs)=\left(\bI_3+\frac{\nabla_{\br}\nabla_{\br}^{\mathsf T}}{\kappa^2}\right)\frac{e^{\im\kappa\|\br-\bs\|}}{4\pi\|\br-\bs\|}\in\mathbb{C}^{3\times 3}
\end{equation}
is the dyadic Green function,\footnote{In \cite{WangCAPA2025} the symbol $\bG$ denotes $-\im\kappa Z_0$ times \eqref{eq:green}, i.e., the field-per-unit-current map with the electric constants absorbed. We keep the bare dyadic \eqref{eq:green} and collect the constants in the impedance kernel \eqref{eq:Zkernel}, so that the correspondence between the two conventions is $\bG_{\text{\cite{WangCAPA2025}}}=-\im\kappa Z_0\,\bG$.} depending on $(\br,\bs)$ only through $\boldsymbol{\tau}=\br-\bs$, so we freely write $\bG(\boldsymbol{\tau})$. Its closed forms, and the resistive-reactive split that organizes the entire paper, are collected once and for all.

\begin{lemma}[Closed forms and split of the dyadic kernel]\label{lem:closed}
With $R=\|\boldsymbol{\tau}\|$, $\hat{\boldsymbol{\tau}}=\boldsymbol{\tau}/R$, $x=\kappa R$, and $g(R)=e^{\im\kappa R}/(4\pi R)$,
\begin{align}
\bG(\boldsymbol{\tau})&=a(R)\,\bI_3+b(R)\,\hat{\boldsymbol{\tau}}\hat{\boldsymbol{\tau}}^{\mathsf T}\label{eq:greenclosed}\\
&=\frac{\im\kappa}{4\pi}\Big[\Big(h_0(x)-\frac{h_1(x)}{x}\Big)\bI_3+h_2(x)\,\hat{\boldsymbol{\tau}}\hat{\boldsymbol{\tau}}^{\mathsf T}\Big],\label{eq:hankelform}\\
a(R)&=g(R)\Big(1+\tfrac{\im}{\kappa R}-\tfrac{1}{(\kappa R)^2}\Big),\nonumber\\
b(R)&=-g(R)\Big(1+\tfrac{3\im}{\kappa R}-\tfrac{3}{(\kappa R)^2}\Big),\nonumber
\end{align}
and, splitting $h_\ell=j_\ell+\im y_\ell$,
\begin{align}
\Imag\{\bG(\boldsymbol{\tau})\}&=\frac{\kappa}{4\pi}\Big[\Big(j_0(x)-\frac{j_1(x)}{x}\Big)\bI_3+j_2(x)\,\hat{\boldsymbol{\tau}}\hat{\boldsymbol{\tau}}^{\mathsf T}\Big],\label{eq:imG}\\
\Real\{\bG(\boldsymbol{\tau})\}&=-\frac{\kappa}{4\pi}\Big[\Big(y_0(x)-\frac{y_1(x)}{x}\Big)\bI_3+y_2(x)\,\hat{\boldsymbol{\tau}}\hat{\boldsymbol{\tau}}^{\mathsf T}\Big].\label{eq:reG}
\end{align}
The imaginary part is entire, with $\Imag\{\bG(\mathbf{0})\}=\tfrac{\kappa}{6\pi}\bI_3$. The real part diverges as $R^{-1}$ through $y_0$ and as $R^{-3}$ through $y_1/x$ and $y_2$.
\end{lemma}
\begin{IEEEproof}
Please refer to Appendix~\ref{app:closed}.
\end{IEEEproof}

Two consequences of Lemma~\ref{lem:closed} are used throughout. As $\kappa R\to\infty$, $a\to g$ and $b\to-g$: each gradient in \eqref{eq:green} acts as multiplication by $\im\kappa\hat{\boldsymbol{\tau}}$, and $\bG$ collapses to the transverse projector $g(R)(\bI_3-\hat{\boldsymbol{\tau}}\hat{\boldsymbol{\tau}}^{\mathsf T})$: the far field carries no component along the propagation direction. And the bounded/singular asymmetry between \eqref{eq:imG} and \eqref{eq:reG} will separate radiation from reaction at every level of the model.\footnote{At coincident arguments $\Real\{\bG\}$ is understood in the principal-value sense with the standard depolarization-dyadic extraction. No quantity below evaluates $\Real\{\bG\}$ at $\boldsymbol{\tau}=\mathbf{0}$.}

User $k$ is a dipole with unit polarization $\boldsymbol{\psi}_k\in\mathbb{C}^3$, $\|\boldsymbol{\psi}_k\|=1$, its extent being much smaller than the source. It observes
\begin{align}
y_k&=\boldsymbol{\psi}_k^{\mathsf H}\boldsymbol{\mathcal{E}}(\br_k)+n_k=\int_{\surf} \mathbf{g}_k^{\mathsf H}(\bs)\,\bj(\bs)\,\dd\bs+n_k,\label{eq:rx}\\
\mathbf{g}_k^{\mathsf H}(\bs)&\triangleq \im\kappa Z_0\,\boldsymbol{\psi}_k^{\mathsf H}\bG(\br_k,\bs)\in\mathbb{C}^{1\times 3},\label{eq:contchan}
\end{align}
with $n_k\sim\mathcal{CN}(0,\sigma_k^2)$. The framework is agnostic with respect to the propagator, and a kernel accounting for fading or reflections can replace \eqref{eq:green} in the continuous channel $\mathbf{g}_k$, exactly as in operator-based holographic designs \cite{IacovelliHMISO}.

\subsection{The impedance kernel and complex power}\label{sec:coupling}
Idealized holographic and FAS models constrain the current norm, $\int_\surf\|\bj\|^2\dd\bs\le p$: analytically convenient, but physically agnostic \cite[Rem.~2]{WangCAPA2025}, because the supplied power is set by the fields the current works against, which are nonlocal functionals of the current itself.

Two such fields act, and they must not be confused. The radiated field \eqref{eq:field}, evaluated back on $\surf$, opposes the sources (back electromotive force: hence the minus sign below). The internal field $\boldsymbol{\mathcal{E}}_{\rm in}=Z_s\bj$ lives inside the conductor: modeling the aperture as a thin sheet of conductivity $\sigma$ and thickness $t$, Ohm's law requires it to push the current through the lossy sheet, with $Z_s\triangleq1/(\sigma t)\ge0$ the surface resistance. Since products of phasors time-average as $\langle\Real\{a e^{-\im\omega t}\}\Real\{b e^{-\im\omega t}\}\rangle=\tfrac12\Real\{ab^{*}\}$, retaining also the imaginary part, which tracks energy exchanged with zero net transfer, defines the complex power the generators deliver:
\begin{align}
P_{\rm cx}&=\frac{1}{2}\int_\surf \bj^{\mathsf H}\boldsymbol{\mathcal{E}}_{\rm in}\,\dd\bs-\frac{1}{2}\int_\surf \bj^{\mathsf H}\boldsymbol{\mathcal{E}}\,\dd\bs
=\frac{1}{2}\big\langle \bj,\,\kerZ\,\bj\big\rangle,\label{eq:cplxpower}\\
\kerZ(\bs,\bs')&\triangleq Z_s\,\delta(\bar{\bs}-\bar{\bs}')\,\bI_3-\im\kappa Z_0\,\bG(\bs,\bs')\quad[\Omega/\text{m}^2],\label{eq:Zkernel}
\end{align}
where the impedance kernel \eqref{eq:Zkernel} maps the current density into the voltage density each point must be driven with: the continuum version of $\tfrac12VI^{*}$, and precisely the coupling kernel of \cite{WangCAPA2025}.

The kernel splits along the two parts of $\bG$, with a twist worth flagging: the prefactor $-\im$ exchanges their roles, $-\im(\Real\{\bG\}+\im\Imag\{\bG\})=\Imag\{\bG\}-\im\Real\{\bG\}$, so the imaginary part of $\bG$ carries the resistive physics and the ohmic delta joins it:
\begin{equation}\label{eq:RXsplit}
\kerZ=\kerR+\im\,\kerX,\quad
\begin{cases}
\kerR(\boldsymbol{\tau})=Z_s\,\delta(\bar{\boldsymbol{\tau}})\,\bI_3+\kappa Z_0\,\Imag\{\bG(\boldsymbol{\tau})\},\\[2pt]
\kerX(\boldsymbol{\tau})=-\kappa Z_0\,\Real\{\bG(\boldsymbol{\tau})\}.
\end{cases}
\end{equation}
By Lemma~\ref{lem:closed} both kernels are real and symmetric, so each quadratic form is individually real (conjugate and relabel variables), and the real part of $P_{\rm cx}$ is carried by $\kerR$ exactly:
\begin{align}
P_{\rm cx}&=P_{\rm em}+\im\, Q,\nonumber\\
P_{\rm em}&=\tfrac12\langle\bj,\kerR\,\bj\rangle=P_\Omega+P_{\rm rad},\qquad Q=\tfrac12\langle\bj,\kerX\,\bj\rangle.\label{eq:realreactive}
\end{align}
So far these are names attached to quadratic forms. The physics is the content of the following.

\begin{proposition}[Power balance and physical identification]\label{prop:power}
Let $\bj$ radiate $\boldsymbol{\mathcal{E}}$, $\boldsymbol{\mathcal{H}}$ in free space according to \eqref{eq:field}, let $\tilde{\bj}(\mathbf{k})=\int_\surf\bj(\bs)e^{-\im\mathbf{k}\cdot\bs}\dd\bs$ be the current spectrum, and let the aperture conductor be the thin sheet above ($t\to0$ at fixed $Z_s$). Then
\begin{equation}\label{eq:poynting}
P_{\rm cx}=P_\Omega+P_{\rm rad}+\im\,2\omega\big(W_e-W_m\big),
\end{equation}
with the identifications:

\!(i) the radiation kernel admits the plane-wave decomposition
\begin{equation}\label{eq:planewave}
\Imag\{\bG(\boldsymbol{\tau})\}=\frac{\kappa}{16\pi^2}\int_{\mathbb{S}^2}\big(\bI_3-\hat{\mathbf{k}}\hat{\mathbf{k}}^{\mathsf T}\big)\,e^{\im\kappa\hat{\mathbf{k}}\cdot\boldsymbol{\tau}}\,\dd\Omega(\hat{\mathbf{k}}),
\end{equation}
and consequently the radiative form equals the flux through any sufficiently large sphere, real in the limit, and the energy of the transverse current spectrum on the sphere of radius $\kappa$:
\begin{align}
P_{\rm rad}&=\frac{\kappa Z_0}{2}\big\langle\bj,\Imag\{\bG\}\bj\big\rangle
=\lim_{r_0\to\infty}\frac12\oint_{r=r_0}\!\big(\boldsymbol{\mathcal{E}}\times\boldsymbol{\mathcal{H}}^{*}\big)\cdot\hat{\mathbf{n}}\,\dd S\nonumber\\
&=\frac{\kappa^2 Z_0}{32\pi^2}\int_{\mathbb{S}^2}\big\|\big(\bI_3-\hat{\mathbf{k}}\hat{\mathbf{k}}^{\mathsf T}\big)\tilde{\bj}(\kappa\hat{\mathbf{k}})\big\|^2\dd\Omega\;\ge0.\label{eq:Pradforms}
\end{align}

(ii) $P_\Omega=\tfrac{Z_s}{2}\int_\surf\|\bj\|^2\dd\bs$ is the sheet limit of the volumetric Joule dissipation $\tfrac12\int\sigma\|\boldsymbol{\mathcal{E}}_{\rm in}\|^2\dd V$.

(iii) $Q=\tfrac12\langle\bj,\kerX\bj\rangle=2\omega(W_e-W_m)$, with $W_e$, $W_m$ the time-average stored electric and magnetic energies, individually divergent with the ball radius but with a finite difference.
\end{proposition}
\begin{IEEEproof}
Please refer to Appendix~\ref{app:poynting}.
\end{IEEEproof}

The reading of Proposition~\ref{prop:power} dictates where each kernel enters the design. The resistive kernel $\kerR$ prices average power. It is bounded away from the ohmic delta and positive semidefinite by the last member of \eqref{eq:Pradforms}: an integral of squared norms, the far-field intensity \cite[App.~A]{WangCAPA2025}. All three representations will serve: the kernel form in the constraints, the spectral form in Section~\ref{sec:wavenumber}, the flux form as physical anchor. The reactive kernel $\kerX$ contributes no average power but computes, from aperture quantities alone, the stored imbalance the generators circulate every half-cycle (the Green function has pre-integrated the near field). Singular as $(\kappa R)^{-3}$, it determines the port voltages sustaining a given current, hence drive feasibility, matching, and bandwidth, and enters through the voltage constraints of Section~\ref{sec:constraints}.

\subsection{Port model and the circuit-field equivalence}\label{sec:ports}
The two coupling cultures of Table~\ref{tab:sota} keep different books: field theory prices power through a kernel acting on current densities. Circuit theory prices voltages and currents through a matrix acting on feeds \cite{Ivrlac2010}. They meet on the complex power, since both must assign the same $P_{\rm cx}$ to every excitation, and this single requirement determines the matrix from the kernel, closed forms included. Table~\ref{tab:dict} states the resulting dictionary. Theorem~\ref{thm:equiv} proves it. The ports are identical, translated elements: port $a$ carries $\boldsymbol{\xi}_a(\bs)=\boldsymbol{\xi}(\bar{\bs}-\bu_a)$, with common profile $\boldsymbol{\xi}(\bar{\bs})=\xi(\bar{\bs})\,\bp$ real, even, tangentially polarized along a fixed real unit $\bp\perp\hat{\mathbf{z}}$, supported on a set of diameter $\Delta\ll\lambda$, and normalized to $\int\xi\,\dd\bar{\bs}=\ell$ [m]. The array current is $\bj=\sum_{a}\iota_a\boldsymbol{\xi}_a$ with feed currents $\boldsymbol{\iota}\in\mathbb{C}^A$.

\begin{table}[!t]
\centering
\caption{The field-circuit dictionary established by Theorem~\ref{thm:equiv}.}
\label{tab:dict}
\footnotesize
\begin{tabular}{@{}ll@{}}
\toprule
Field theory (kernels) & Circuit theory (matrices) \\
\midrule
current density $\bj(\bs)$ & feed currents $\boldsymbol{\iota}$ \\
impedance kernel $\kerZ$ in \eqref{eq:cplxpower} & $\bZ(\Utt)$: pullback \eqref{eq:portZ} \\
voltage density $(\kerZ\bj)(\bs)$ & port voltages $\mathbf{v}=\bZ(\Utt)\boldsymbol{\iota}$ \\
$P_{\rm cx}=\tfrac12\langle\bj,\kerZ\bj\rangle$ & $P_{\rm cx}=\tfrac12\boldsymbol{\iota}^{\mathsf H}\bZ(\Utt)\boldsymbol{\iota}$ \\
resistive kernel $\kerR$: $P_\Omega+P_{\rm rad}$ & $\Real\{\bZ(\Utt)\}$: same powers \\
reactive kernel $\kerX$: $2\omega(W_e-W_m)$ & $\Imag\{\bZ(\Utt)\}$: same imbalance \\
closed forms \eqref{eq:imG}-\eqref{eq:reG} & mutual entries $R_r\vartheta$, \eqref{eq:vartheta} \\
element profile $\boldsymbol{\xi}$ & diagonal $(R_\Omega,\ R_r,\ X_A)$ \\
\bottomrule
\end{tabular}
\end{table}
\newpage
\begin{theorem}[Circuit-field equivalence]\label{thm:equiv}
The following hold.

(i) (One power, one matrix.) Substituting the port current into \eqref{eq:cplxpower} gives $P_{\rm cx}=\tfrac12\,\boldsymbol{\iota}^{\mathsf H}\bZ(\Utt)\,\boldsymbol{\iota}$ with
\begin{align}
[\bZ(\Utt)]_{ab}&=\langle\boldsymbol{\xi}_a,\kerZ\,\boldsymbol{\xi}_b\rangle=(\xi\star\zeta_\bp\star\xi)(\bu_a-\bu_b),\label{eq:portZ}\\
\zeta_\bp(\bar{\boldsymbol{\tau}})&\triangleq Z_s\,\delta(\bar{\boldsymbol{\tau}})-\im\kappa Z_0\,\bp^{\mathsf H}\bG(\bar{\boldsymbol{\tau}})\bp,\nonumber
\end{align}
$\star$ denoting planar correlation. This complex symmetric matrix is the multiport impedance matrix of circuit theory: matching complex power for every drive determines it uniquely, and its entries recover the induced-EMF mutual impedance.

(ii) (The split descends.) $\Real\{\bZ(\Utt)\}$ and $\Imag\{\bZ(\Utt)\}$ are the port-sampled resistive and reactive kernels of \eqref{eq:RXsplit}, so Proposition~\ref{prop:power} holds verbatim at port level:
$\tfrac12\boldsymbol{\iota}^{\mathsf H}\Real\{\bZ\}\boldsymbol{\iota}=P_\Omega+P_{\rm rad}$ and $\tfrac12\boldsymbol{\iota}^{\mathsf H}\Imag\{\bZ\}\boldsymbol{\iota}=2\omega(W_e-W_m)$.\label{eq:portpower}

(iii) (Common closed form.) As $\kappa\Delta\to0$, the mutual entries ($a\ne b$, $\|\bu_a-\bu_b\|\ge\Delta$) converge to the kernel sampled at the separations: $[\bZ]_{ab}\to R_r\,\vartheta(\bu_a-\bu_b)$, $R_r=Z_0\kappa^2\ell^2/(6\pi)$, with, for $x=\kappa\|\bar{\boldsymbol{\tau}}\|$ and $c=\bp^{\mathsf T}\hat{\boldsymbol{\tau}}$,
\begin{equation}\label{eq:vartheta}
\vartheta(\bar{\boldsymbol{\tau}})\triangleq\rho(\bar{\boldsymbol{\tau}})+\im\chi(\bar{\boldsymbol{\tau}})=\frac{3}{2}\Big[h_0(x)-\frac{h_1(x)}{x}+h_2(x)\,c^2\Big],
\end{equation}
whose real part $\rho$ ($j_\ell$ family) is the normalized resistive kernel \eqref{eq:imG} and whose imaginary part $\chi$ ($y_\ell$ family) is the normalized reactive kernel \eqref{eq:reG}. The radiation self-resistance also admits the point limit, equal to the Hertzian $R_r$.

(iv) (Profile-dependent diagonal.) The self-reactance $X_A\triangleq\Imag\{[\bZ]_{aa}\}$ and the ohmic self-resistance $R_\Omega\triangleq Z_s\int|\xi|^2$ are finite for every fixed profile but diverge as $\kappa\Delta\to0$ (as $(\kappa\Delta)^{-3}$ and $\Delta^{-2}$): they are irreducibly properties of the element.
Consequently $\bZ(\Utt)=(R_\Omega+R_r+\im X_A)\bI_A+R_r[\vartheta(\bu_a-\bu_b)]_{a\ne b}+O((\kappa\Delta)^2)$ on the mutual terms.
\end{theorem}
\begin{IEEEproof}
Please refer to Appendix~\ref{app:equiv}.
\end{IEEEproof}

The two cultures are thus one bookkeeping at two resolutions: off the diagonal the circuit matrix is the field kernel sampled at the port separations, with resistance and reactance inheriting \eqref{eq:imG}-\eqref{eq:reG} through the single function \eqref{eq:vartheta}, and $\Real\{\bZ(\Utt)\}$ recovering the discrete coupling matrix of \cite[Sec.~V]{WangCAPA2025}, while on the diagonal field theory proves that no universal value exists, the element constants of the circuit tradition (e.g., $73.08+\im42.21\,\Omega$ \cite{New2024MIMOFAS}) being the missing profile integrals.

\begin{remark}[Attaching generators]\label{rem:cultures}
Throughout, the ports are current-driven and $\boldsymbol{\iota}$ is the design variable. A feed with source voltages $\mathbf{v}_s$ and load $Z_L$ merely re-parametrizes it, $\boldsymbol{\iota}=(\bZ(\Utt)+Z_L\bI_A)^{-1}\mathbf{v}_s$: the fixed-load FAS model of \cite[App.~III]{New2024MIMOFAS} is exactly this map, a corollary of Theorem~\ref{thm:equiv}, its drive-side cost captured by the voltage constraint of Section~\ref{sec:constraints}.
\end{remark}

\begin{remark}[Scalar reduction and the $\lambda/2$ intuition]\label{rem:sinc}
Retaining only the scalar Green function collapses $\rho$ to $\rho_0(x)=j_0(x)=\sinc(x)$, with zeros at multiples of $\lambda/2$: the precise sense in which classical MIMO sees no coupling on that grid. The polarized kernel \eqref{eq:vartheta} shifts these nulls, quantified in Section~\ref{sec:fasequiv} through the null equations of \cite{WangCAPA2025}.
\end{remark}

\subsection{Multi-user signal and coupled constraints}\label{sec:constraints}
The multi-user transmit signal superposes $K$ precoded streams on the $A$ ports, $\boldsymbol{\iota}=\sum_{k=1}^{K}\bw_k x_k$ with independent zero-mean symbols $x_k$ of unit average power and precoders $\bw_k\in\mathbb{C}^A$ in feed-current units, the complex-baseband envelopes of the analog port currents, with one RF chain per active port. Normalizing $\bZ(\Utt)$ by $R_r$ and writing $\varepsilon\triangleq R_\Omega/R_r$ (the loss-to-radiation ratio. The element efficiency is $1/(1+\varepsilon)$) and $\zeta_A\triangleq X_A/R_r$, Theorem~\ref{thm:equiv} gives
\begin{align}
\bZ(\Utt)&=R_r\Big[\big(1+\varepsilon+\im\zeta_A\big)\bI_A+\boldsymbol{\Theta}(\Utt)\Big],\nonumber\\
[\boldsymbol{\Theta}]_{ab}&=\begin{cases}\vartheta(\bu_a-\bu_b), & a\ne b,\\ 0, & a=b.\end{cases}\label{eq:Znorm}
\end{align}
Two constraints follow from Section~\ref{sec:coupling}. The average-power constraint involves only the resistive part: with $\bC_\varepsilon(\Utt)\triangleq \varepsilon\bI_A+\bC(\Utt)$ and $[\bC(\Utt)]_{ab}=\rho(\bu_a-\bu_b)$ (so $[\bC]_{aa}=\rho(\mathbf{0})=1$: the unit radiation self-resistance sits on the diagonal of $\bC$, and $\bC_\varepsilon$ carries $1+\varepsilon=(R_\Omega+R_r)/R_r$ there, matching the real diagonal of \eqref{eq:Znorm}),
\begin{equation}\label{eq:powcon}
\sum_{k=1}^{K}\mathbb{E}\big\{P_{\rm em}\big\}=\frac{R_r}{2}\sum_{k=1}^{K}\bw_k^{\mathsf H}\,\bC_\varepsilon(\Utt)\,\bw_k\le P.
\end{equation}
The voltage-feasibility constraint involves the full matrix \eqref{eq:Znorm}: the port voltages are $\mathbf{v}=\bZ(\Utt)\boldsymbol{\iota}$, so, per port $a$ and with $\mathbf{e}_a$ the $a$-th canonical vector,
\begin{equation}\label{eq:voltcon}
\mathbb{E}\big\{|v_a|^2\big\}=\sum_{k=1}^{K}\big|\mathbf{e}_a^{\mathsf T}\bZ(\Utt)\bw_k\big|^2\le V_{\max}^2,\qquad a=1,\dots,A,
\end{equation}
which is where the reactive kernel, dominant at small spacings by \eqref{eq:reG}, bites. Equations \eqref{eq:rx}, \eqref{eq:powcon}, and \eqref{eq:voltcon} constitute the coupling-consistent multi-user model. Note where the coupling lives: not in the channel map \eqref{eq:rx}, which is linear in the currents, but in the geometry of the feasible set. Two configurations that would be equivalent under a norm constraint may differ vastly in radiated power and required drive voltages. Conversely, a configuration may radiate far less than its norm suggests, which is precisely the loophole superdirectivity exploits, now visibly taxed by $\varepsilon$ and $V_{\max}$.

The model of this section is complete but spatial. Its structure, a band-limited channel functional \eqref{eq:rx} and shift-invariant kernels \eqref{eq:RXsplit}, begs for a Fourier treatment, and the next section provides it.

\section{Wavenumber-Domain Representation and the FAS-HMIMO Equivalence}\label{sec:wavenumber}
This section carries the model of Section~\ref{sec:model} into the wavenumber domain, where the split \eqref{eq:RXsplit} becomes an exact statement about spectral supports, and then shows that the flexible-position antenna lives there as a simple object: a constant-modulus codeword on the wavenumber lattice.

\subsection{The light circle splits the impedance kernel}\label{sec:split}
For tangential currents on the aperture plane, $\bj(\bs)=\bp\,j(\bar{\bs})$ with the fixed real $\bp\perp\hat{\mathbf{z}}$ of Section~\ref{sec:ports}, the relevant restriction of the impedance kernel is the scalar function $\zeta_\bp(\bar{\boldsymbol{\tau}})$ of Theorem~\ref{thm:equiv}. Its planar Fourier transform is where radiation and reaction part ways.

\begin{proposition}[Light-circle dichotomy]\label{prop:split}
In the sense of tempered distributions, $\zeta_\bp(\bar{\boldsymbol{\tau}})=\tfrac{1}{4\pi^2}\int_{\mathbb{R}^2}\widehat{\zeta}_\bp(\bk)e^{\im\bk\cdot\bar{\boldsymbol{\tau}}}\dd\bk$ with
\begin{equation}\label{eq:zspec}
\widehat{\zeta}_\bp(\bk)=Z_s+\kappa Z_0\,\frac{1-(\bk\cdot\bar{\bp})^2/\kappa^2}{2\,k_z},\qquad k_z=\sqrt{\kappa^2-\|\bk\|^2},
\end{equation}
where the square root carries $\Imag\{k_z\}\ge0$. Consequently, the resistive and reactive spectra separate exactly at the light circle $\|\bk\|=\kappa$:
\begin{align}
\widehat{r}_\bp(\bk)&\triangleq\Real\{\widehat{\zeta}_\bp(\bk)\}\!=\!Z_s\!+\!\kappa Z_0\,\frac{1-(\bk\cdot\bar{\bp})^2/\kappa^2}{2k_z}\,\mathsf{1}\{\|\bk\|<\kappa\},\label{eq:rspec}\\
\widehat{x}_\bp(\bk)&\triangleq\Imag\{\widehat{\zeta}_\bp(\bk)\}\!=\!-\kappa Z_0\,\frac{1-(\bk\cdot\bar{\bp})^2/\kappa^2}{2|k_z|}\,\mathsf{1}\{\|\bk\|>\kappa\}.\label{eq:xspec}
\end{align}
The resistive spectrum is the ohmic floor $Z_s$ plus the obliquity-and-polarization weight on the visible disc, in agreement with the CAPA coupling kernel of \cite[Eq.~(33)]{WangCAPA2025}. The reactive spectrum lives entirely in the evanescent region, changes sign across the polarization cone $|\bk\cdot\bar{\bp}|=\kappa$ (capacitive versus inductive storage), and grows linearly in $\|\bk\|$, mirroring the $R^{-3}$ singularity of $\kerX$.
\end{proposition}
\begin{IEEEproof}
Please refer to Appendix~\ref{app:spectral}.
\end{IEEEproof}

Proposition~\ref{prop:split} upgrades the informal statement that evanescent waves do not radiate to an exact operator identity, and makes the average power computable by inspection. Since average power is all that the transmit budget prices, attention restricts to the resistive part of the kernel: define the scalar resistive kernel
\begin{equation}\label{eq:rpdef}
r_\bp(\bar{\boldsymbol{\tau}})\triangleq\Real\{\zeta_\bp(\bar{\boldsymbol{\tau}})\}=Z_s\,\delta(\bar{\boldsymbol{\tau}})+\underbrace{\kappa Z_0\,\bp^{\mathsf H}\Imag\{\bG(\bar{\boldsymbol{\tau}})\}\bp}_{\triangleq\,r^{\rm rad}_\bp(\bar{\boldsymbol{\tau}})},
\end{equation}
whose transform is \eqref{eq:rspec}. Substituting this spectral representation into the quadratic form, recognizing $\int j^{*}(\bar{\bs})e^{\im\bk\cdot\bar{\bs}}\dd\bar{\bs}=\tilde{j}^{*}(\bk)$, and exchanging the integrals yields the Parseval identity
\begin{equation}\label{eq:parseval}
\iint j^{*}(\bar{\bs})\,r_\bp(\bar{\bs}-\bar{\bs}')\,j(\bar{\bs}')\,\dd\bar{\bs}\,\dd\bar{\bs}'
=\frac{1}{4\pi^2}\int_{\mathbb{R}^2}\widehat{r}_\bp(\bk)\,|\tilde{j}(\bk)|^2\,\dd\bk.
\end{equation}
the exchange is legitimate because $\tilde j$ has finite energy and $\widehat{r}_\bp$ is integrable on its support, the rim divergence being harmless since $\int_{\|\bk\|<\kappa}k_z^{-1}\dd\bk=2\pi\kappa$. Halving \eqref{eq:parseval},
\begin{align}
P_{\rm em}&=\frac{1}{8\pi^2}\int_{\mathbb{R}^2}\widehat{r}_\bp(\bk)\,|\tilde{j}(\bk)|^2\,\dd\bk=\frac{Z_s}{8\pi^2}\int_{\mathbb{R}^2}|\tilde{j}(\bk)|^2\,\dd\bk\nonumber\\
&+\frac{\kappa Z_0}{16\pi^2}\int_{\|\bk\|<\kappa}\frac{1-(\bk\cdot\bar{\bp})^2/\kappa^2}{k_z}\,|\tilde{j}(\bk)|^2\,\dd\bk,\label{eq:powerdisc}
\end{align}
the first term returning the ohmic $\tfrac{Z_s}{2}\int\|j\|^2$ by Parseval: the ohmic floor taxes every spectral component, the radiation weight only the visible disc. Coupling has not disappeared, it is the non-flatness of the per-wavenumber power price in \eqref{eq:powerdisc}.

\subsection{Fourier modes and the modal Gram}\label{sec:diag}
Expand the scalar current over the orthonormal Fourier modes of the surface,
\begin{equation}\label{eq:modes}
j(\bar{\bs})=\sum_{n}q_n\,\phi_n(\bar{\bs}),\qquad \phi_n(\bar{\bs})=\frac{e^{\im\bk_n\cdot\bar{\bs}}}{\sqrt{|\surf|}},
\end{equation}
with lateral wavenumbers on the aperture lattice, $\bk_n=2\pi[n_x/S_x,\,n_y/S_y]^{\mathsf T}$, $\mathbf{n}\in\mathbb{Z}^2$, and $\int_\surf\phi_n^{*}\phi_{n'}\dd\bar{\bs}=\delta_{nn'}$. The complex power becomes $P_{\rm cx}=\tfrac12\,\bq^{\mathsf H}(\bB_R+\im\bB_X)\,\bq$ with modal Grams $[\bB_R]_{nm}=\langle\phi_n\bp,\kerR\,\phi_m\bp\rangle$ and likewise for $\bB_X$. The ohmic part of $\bB_R$ is computed exactly by orthonormality: $Z_s\int\phi_n^{*}\phi_m=Z_s\delta_{nm}$. For the radiative part, insert the spectral representation of $r^{\rm rad}_\bp$ from \eqref{eq:rpdef} and factorize the exponential of the difference $\bar{\bs}-\bar{\bs}'$:
\begin{align}
&[\bB_R]_{nm}\!-\!Z_s\delta_{nm}
\!=\!\frac{1}{|\surf|}\!\iint_{\surf^2}\!\!e^{-\im\bk_n\cdot\bar{\bs}}\,r^{\rm rad}_\bp(\bar{\bs}-\bar{\bs}')\,e^{\im\bk_m\cdot\bar{\bs}'}\!\dd\bar{\bs}\,\dd\bar{\bs}'\nonumber\\
&=\frac{1}{4\pi^2|\surf|}\int_{\|\bk\|<\kappa}\widehat{r}^{\,\rm rad}_\bp(\bk)\,D(\bk-\bk_n)\,D(\bk-\bk_m)\,\dd\bk,\label{eq:Bexact}
\end{align}
where the two aperture integrals produced the real, even Dirichlet kernel $D(\bk)\triangleq\int_\surf e^{\im\bk\cdot\bar{\bs}}\dd\bar{\bs}=S_xS_y\,\sinc(k_xS_x/2)\,\sinc(k_yS_y/2)$. Expression \eqref{eq:Bexact} is exact for any aperture size. Its asymptotics are the content of the next result, and the accuracy at finite aperture is quantified numerically in Section~\ref{sec:numerics}.

\begin{lemma}[Asymptotic diagonalization of the resistive Gram]\label{lem:diag}
Fix $\delta>0$ and lattice modes with $\|\bk_n\|,\|\bk_m\|\le\kappa-\delta$. Then, as $\kappa S_x,\kappa S_y\to\infty$,
\begin{equation}\label{eq:Bdiag}
[\bB_R]_{nm}=\widehat{r}_\bp(\bk_n)\,\delta_{nm}+\varepsilon_{nm},\qquad \varepsilon_{nm}\to 0,
\end{equation}
with $|\varepsilon_{nm}|\le \omega\big(S_{\min}^{-1/2}\big)+O\big(S_{\min}^{-q}\big)$ for every $q<1/2$, where $\omega$ is the modulus of continuity of $\widehat{r}_\bp$ on the disc of radius $\kappa-\delta/2$ and $S_{\min}=\min(S_x,S_y)$. Evanescent modes, $\|\bk_n\|\ge\kappa+\delta$, satisfy $[\bB_R]_{nn}\to Z_s$: they dissipate but do not radiate.
\end{lemma}
\begin{IEEEproof}
Please refer to Appendix~\ref{app:lemdiag}.
\end{IEEEproof}

Fourier modes thus asymptotically decouple the aperture, each propagating mode radiating at the individual price $\widehat{r}_\bp(\bk_n)$: the ohmic floor plus a weight minimal near broadside and integrably divergent toward the rim. The reactive Gram $\bB_X$ behaves differently, and \eqref{eq:xspec} says why: $\widehat{x}_\bp$ grows linearly in $\|\bk\|$ while the transforms of the truncated exponentials \eqref{eq:modes} decay only as $1/k$ per axis, so the form $\langle\phi_n\bp,\kerX\phi_n\bp\rangle$ diverges logarithmically. The divergence belongs to the modes rather than to the framework: a Fourier current chopped at the aperture edge carries a line discontinuity, to which the electromagnetic edge condition assigns infinite stored energy. Reactive bookkeeping is accordingly carried at port level, where the profiles are smooth, consistently with Theorem~\ref{thm:equiv}(iv). A caution on terminology: the limit $\kappa S\to\infty$ concerns the aperture relative to the wavelength and involves no user, so no far-field assumption is made anywhere in the paper. Practical users in fact sit inside the Rayleigh distance $2D^2/\lambda$ of an electrically large surface, and the resulting wavefront curvature lives entirely in the spectra $\widehat{g}_{k,n}$, which spread over many visible modes while $\bB_R$ diagonalizes regardless.

\begin{remark}[Edge-compatible modal families]\label{rem:edge}
The truncated exponentials \eqref{eq:modes} are adopted for their translation covariance, exploited in Section~\ref{sec:fasequiv}, but any orthonormal family vanishing at the rim, such as TE/TM-type sinusoidal products on the refined lattice $\pi/S_a$, has transforms decaying as $1/k^2$ per axis and hence finite reactive forms. The radiative analysis transfers by congruence with the same weights (Lemma~\ref{lem:diag}), so edge-compatible families admit wavenumber-domain reactive bookkeeping as well.
\end{remark}

\subsection{Flexible ports as constant-modulus wavenumber codewords}\label{sec:fasequiv}
We can now make precise the folklore statement that ``FAS is HMIMO with antenna selection'' \cite{New2025Oversampling}, and, crucially, make it coupling-aware. Let the antenna consist of the $A$ simultaneously active ports of Section~\ref{sec:ports}: common real profile $\xi$, polarization $\bp$, feed currents $\iota_a$, planar positions $\bu_a\in\surf$,
\begin{equation}\label{eq:fascurrent}
\bj(\bs)=\bp\sum_{a=1}^{A}\iota_a\,\xi(\bar{\bs}-\bu_a).
\end{equation}
Projecting \eqref{eq:fascurrent} on \eqref{eq:modes}, the scalar-current coefficient on mode $n$ follows by orthonormality and a change of variables:
\begin{align}
q_n&=\int_\surf \phi_n^{*}(\bar{\bs})\,\sum_{a=1}^{A}\iota_a\,\xi(\bar{\bs}-\bu_a)\,\dd\bar{\bs}\nonumber\\
&=\frac{1}{\sqrt{|\surf|}}\sum_{a=1}^{A}\iota_a\int_{\mathbb{R}^2} \xi(\bar{\bs}-\bu_a)\,e^{-\im\bk_n\cdot\bar{\bs}}\,\dd\bar{\bs}\nonumber\\
&=\frac{1}{\sqrt{|\surf|}}\sum_{a=1}^{A}\iota_a\,e^{-\im\bk_n\cdot\bu_a}\,\widehat{\xi}(\bk_n)
=\sum_{a=1}^{A}\iota_a\,\widehat{\xi}(\bk_n)\,[\bm{\varphi}(\bu_a)]_n,\label{eq:codeword}
\end{align}
where the second line inserts $\phi_n^{*}=e^{-\im\bk_n\cdot\bar{\bs}}/\sqrt{|\surf|}$ and extends the integration to the plane (each element footprint lies inside $\surf$), and the third substitutes $\bar{\mathbf{v}}=\bar{\bs}-\bu_a$, with $\widehat{\xi}(\bk)\triangleq\int\xi(\bar{\mathbf{v}})e^{-\im\bk\cdot\bar{\mathbf{v}}}\dd\bar{\mathbf{v}}$ real and even, equal to $\ell$ at the origin and flat over the visible disc when $\kappa\Delta\ll1$. Every port thus contributes the fixed taper $\widehat{\xi}(\bk_n)$ times the constant-modulus codeword $\bm{\varphi}(\bu)\triangleq|\surf|^{-1/2}[e^{-\im\bk_n\cdot\bu}]_{n}$. Entrywise, $[\bm{\varphi}(\bu)]_n=\phi_n^{*}(\bu)$, the conjugate sample of the $n$-th mode at the port location.

\begin{theorem}[Positioning-HMIMO equivalence]\label{thm:fas}
The $A$-port flexible-position antenna over $\surf$ is the HMIMO transmitter of Section~\ref{sec:diag} with mode coefficients confined to the manifold
\begin{equation}\label{eq:manifold}
\mathcal{Q}_A=\Big\{\sum_{a=1}^{A}\iota_a\,\widehat{\xi}\odot\bm{\varphi}(\bu_a)\;:\;\iota_a\in\mathbb{C},\;\bu_a\in\surf\Big\},
\end{equation}
with $\odot$ the entrywise product by the taper, and under this embedding:

(i) the received functional of user $k$ superposes codeword evaluations of the Fourier series of its tapered channel,
\begin{align}
\int_\surf\mathbf{g}_k^{\mathsf H}(\bs)\,\bj(\bs)\,\dd\bs&=\sum_{a=1}^{A}\iota_a\,h_k(\bu_a),\nonumber\\
h_k(\bu)&=\sum_n \widehat{g}_{k,n}\,\widehat{\xi}(\bk_n)\,[\bm{\varphi}(\bu)]_n,\label{eq:chanseries}
\end{align}
with $\widehat{g}_{k,n}=\int_\surf \mathbf{g}_k^{\mathsf H}(\bs)\bp\,\phi_n(\bar{\bs})\,\dd\bar{\bs}$: port motion is a spectral modulation, and the FAS skimming of the field of \cite{New2025Oversampling} is codeword evaluation.

(ii) the complex power computed in the mode domain coincides, for every aperture, with the port form of Theorem~\ref{thm:equiv},
\begin{equation}\label{eq:modalexact}
\tfrac12\,\bq^{\mathsf H}\big(\bB_R+\im\bB_X\big)\,\bq=\tfrac12\,\boldsymbol{\iota}^{\mathsf H}\bZ(\Utt)\,\boldsymbol{\iota},
\end{equation}
and its resistive part admits, by Lemma~\ref{lem:diag}, the diagonal-limit representation
\begin{align}
P_{\rm em}\;&\to\;\frac12\sum_{a=1}^{A}\sum_{b=1}^{A}\iota_a^{*}\,\iota_b\,\Psi(\bu_a-\bu_b),\nonumber\\
\Psi(\bar{\boldsymbol{\tau}})\!&\triangleq\!\!\!\lim_{\kappa S\to\infty}\frac{1}{|\surf|}\sum_{n}\widehat{r}_\bp(\bk_n)\,\widehat{\xi}(\bk_n)^2\,e^{\im\bk_n\cdot\bar{\boldsymbol{\tau}}}\!=\!(\xi\star r_\bp\star\xi)(\bar{\boldsymbol{\tau}}),\label{eq:pullback}
\end{align}
the lattice sum being the Riemann sum of the spectral integral \eqref{eq:parseval}: the coupling between ports is the pullback of the wavenumber kernel under the codeword map \eqref{eq:codeword}.
\end{theorem}
\begin{IEEEproof}
(i) By linearity, $\int\mathbf{g}_k^{\mathsf H}\bj=\sum_a\iota_a\int_\surf\mathbf{g}_k^{\mathsf H}(\bs)\bp\,\xi(\bar{\bs}-\bu_a)\,\dd\bar{\bs}$. The single-port case of \eqref{eq:codeword} is the mode expansion of the shifted profile, $\xi(\cdot-\bu)=\sum_n\widehat{\xi}(\bk_n)\,[\bm{\varphi}(\bu)]_n\,\phi_n$ in $L^2(\surf)$, and integrating it against $\mathbf{g}_k^{\mathsf H}\bp\in L^2(\surf)$ term by term (Parseval) produces \eqref{eq:chanseries}.
(ii) Exactness: the two sides of \eqref{eq:modalexact} are the mode-basis and port-basis evaluations of the same quadratic form $\tfrac12\langle\bj,\kerZ\bj\rangle$ of \eqref{eq:cplxpower}, equal by Parseval. Limit: Lemma~\ref{lem:diag} replaces $\bB_R$ by $\diag\{\widehat{r}_\bp(\bk_n)\}$. Substituting \eqref{eq:codeword} produces the double port sum over the lattice sum in \eqref{eq:pullback}, whose cell area $4\pi^2/|\surf|$ identifies it as the Riemann sum of $\tfrac{1}{4\pi^2}\int\widehat{r}_\bp(\bk)\,\widehat{\xi}(\bk)^2e^{\im\bk\cdot\bar{\boldsymbol{\tau}}}\dd\bk=(\xi\star r_\bp\star\xi)(\bar{\boldsymbol{\tau}})$, convergent by the rim integrability noted after \eqref{eq:parseval} and by the decay of $\widehat{\xi}$. The same decay keeps the ohmic and reactive lattice sums finite, consistently with Theorem~\ref{thm:equiv}(iv), whereas for $\widehat{\xi}\equiv\ell$ the ohmic sum diverges.
\end{IEEEproof}

Theorem~\ref{thm:fas} supplies the dictionary this paper is built on. Positioning is modulation: the position enters only through the phase ramp $\bm{\varphi}(\bu)$. Selection is codebook restriction: the attainable coefficient vectors are exactly the image of the admissible positions under $\bu\mapsto\widehat{\xi}\odot\bm{\varphi}(\bu)$, so constraining where the ports may sit constrains which codewords exist. The full surface yields the continuous manifold \eqref{eq:manifold}, a subregion $\mathcal{R}\subset\surf$ its sub-manifold, and a candidate grid $\mathcal{G}$ the finite codebook $\{\widehat{\xi}\odot\bm{\varphi}(\bu_g)\}_{\bu_g\in\mathcal{G}}$, so activating $A$ ports picks $A$ codewords: port selection is codeword selection. The price of activating several codewords is not additive. It is the quadratic form \eqref{eq:pullback}, mutual coupling, resistive and reactive alike.

\begin{remark}[Movable and fluid antennas unified]\label{rem:unif}
The two architectures of the flexible-positioning literature are the two position sets of Theorem~\ref{thm:fas}: continuous repositioning (MA \cite{ZhuMA2024TWC}) optimizes over the full manifold \eqref{eq:manifold}, while port selection (FAS \cite{New2024MIMOFAS}) restricts it to a grid's finite codebook. Every statement of this paper, kernels, constraints, gradients, and bounds, applies verbatim to both, and the algorithm of Section~\ref{sec:opt} mirrors the split: projected-gradient ascent for the former, FFT-greedy codeword selection for the latter.
\end{remark}

Two consequences follow.

\begin{corollary}[Half-wavelength lattice]\label{cor:halflambda}
Under the scalar reduction $\rho_0=j_0$, ports on the $\lambda/2$ grid satisfy $\bC(\Utt)=\bI_A$: classical uncoupled MIMO is a flexible-position antenna frozen on the half-wavelength lattice. Under the polarized kernel \eqref{eq:vartheta}, the resistive nulls solve $(x^2-1)\sin x+x\cos x=0$ for $c=0$ and $\tan x=x$ for $c=1$, displacing the first nulls to $\approx0.44\lambda$ and $\approx0.72\lambda$ as in \cite{WangCAPA2025}, while at $x=m\pi$ (broadside, $c=0$)
\begin{equation}\label{eq:residual}
\rho=\frac{3(-1)^m}{2\pi^2m^2},\qquad
\chi=\frac{3(-1)^{m+1}}{2}\Big(\frac{1}{m\pi}-\frac{1}{(m\pi)^3}\Big):
\end{equation}
the coupling surviving on the classical grid is first order and predominantly reactive ($\chi\approx0.43$ against $\rho\approx0.15$ at $\lambda/2$), motivating the voltage constraint \eqref{eq:voltcon} even far from the superdirective regime.
\end{corollary}
\begin{IEEEproof}
$j_0(m\pi)=0$ gives the first claim. For $c=0$, $\rho\propto j_0-j_1/x$, and multiplying by $x^3$ gives the first null equation. For $c=1$, $j_2=3j_1/x-j_0$ collapses $\rho\propto 2j_1/x$, whose zeros solve $\tan x=x$. At $x=m\pi$: $j_1=-(-1)^m/x$, $j_2=-3(-1)^m/x^2$, $y_0=-(-1)^m/x$, $y_1=-(-1)^m/x^2$. Substitution in \eqref{eq:vartheta} yields \eqref{eq:residual}.
\end{IEEEproof}

The second consequence explains why single-port FAS analyses could afford to ignore coupling. For $A=1$, the power reduces to $P_{\rm em}=\tfrac{R_r}{2}(1+\varepsilon)|\iota|^2$ independently of the position $\bu$ (up to aperture-edge corrections captured by the finite-$\surf$ Gram \eqref{eq:Bexact}): with one active port there is nothing to couple to. The moment $A\ge 2$, the feasible power and voltage sets become position-dependent through $\bC_\varepsilon(\Utt)$ and $\bZ(\Utt)$, and any optimization that moves ports while radiating from several of them must be coupling-aware. That optimization is the subject of the next section.

\section{Coupling-Aware Sum-Rate Maximization}\label{sec:opt}
Armed with the equivalence of Section~\ref{sec:wavenumber}, we now pose the design problem. The transmitter is the $A$-port flexible-position antenna \eqref{eq:fascurrent}. Each port is shared by the $K$ streams through per-user precoders $\bw_k\in\mathbb{C}^A$ in feed-current units, $\boldsymbol{\iota}=\sum_{k}\bw_k x_k$. By Theorem~\ref{thm:fas}(i), user $k$ observes
\begin{align}
y_k&=\bh_k^{\mathsf H}(\Utt)\sum_{i=1}^{K}\bw_i\,x_i+n_k,\nonumber\\
[\bh_k(\Utt)]_a&= h_k(\bu_a)\ \text{as in \eqref{eq:chanseries}},\label{eq:rxmu}
\end{align}
so that its signal-to-interference-plus-noise ratio is
\begin{equation}\label{eq:sinr}
\Gamma_k(\Utt,\mathbf{W})=\frac{|\bh_k^{\mathsf H}\bw_k|^2}{\sum_{i\ne k}|\bh_k^{\mathsf H}\bw_i|^2+\sigma_k^2},\qquad \mathbf{W}=[\bw_1,\dots,\bw_K].
\end{equation}
Absorbing $R_r/2$ into $P$ and $R_r$ into $V_{\max}$, the master problem collects \eqref{eq:powcon}-\eqref{eq:voltcon}:
\begin{align}
(\mathrm{P}1): \max_{\Utt,\,\mathbf{W}}\ &\ \sum_{k=1}^{K}\alpha_k\log_2\big(1+\Gamma_k(\Utt,\mathbf{W})\big)\label{eq:P1}\\
\text{s.t.}\ &\ \mathrm{C}_1:\ \textstyle\sum_{k}\bw_k^{\mathsf H}\,\bC_\varepsilon(\Utt)\,\bw_k\le P,\nonumber\\
&\ \mathrm{C}_2:\ \textstyle\sum_{k}\big|\mathbf{e}_a^{\mathsf T}\bar{\bZ}(\Utt)\bw_k\big|^2\le \bar V_{\max}^2,\ a=1,\dots,A,\nonumber\\
&\ \mathrm{C}_3:\ \bu_a\in\surf,\quad \|\bu_a-\bu_b\|\ge d_{\min},\ \ a\ne b,\nonumber
\end{align}
where $\alpha_k\ge0$ are rate weights, $\bar\bZ=\bZ/R_r$ is the normalized \eqref{eq:Znorm}, and $d_{\min}\ge\Delta$ is the hardware-exclusion radius of the elements. In the wavenumber picture, (P1) is a mode-codeword selection problem: the optimization ranges over $K$ precoded superpositions of $A$ constant-modulus codewords \eqref{eq:codeword} drawn from the manifold \eqref{eq:manifold}, with the spectral weights of Proposition~\ref{prop:split} pricing every configuration. Positioning and mode selection are the same variable seen in two domains, and Section~\ref{sec:wnimpl} turns the correspondence into the data structure the algorithm runs on. Restricting $\bu_a$ to a candidate grid turns (P1) into codebook selection, handled by the FFT machinery of Section~\ref{sec:wnimpl}, while we solve the continuous form.

Problem (P1) is non-convex on three counts: the sum rate in $\mathbf{W}$, the trigonometric dependence on $\Utt$ through both $\bh_k$ and the kernels, and their multiplication in the constraints. We adopt an alternating strategy whose two blocks are individually well understood, and whose novelty resides in what the blocks see: both the channel and the feasible set move with $\Utt$. The voltage constraint $\mathrm{C}_2$ is handled inside each block below.

\subsection{Coupling whitening and the precoder block}\label{sec:wmmse}
Fix $\Utt$. The matrix $\bC_\varepsilon(\Utt)=\varepsilon\bI_A+\bC(\Utt)$ is positive definite for every configuration: $\bC$ is positive semidefinite as the pullback of the positive-semidefinite radiation kernel (Section~\ref{sec:coupling}), and the ohmic floor shifts the spectrum by $\varepsilon>0$. Define the whitened variables
\begin{equation}\label{eq:whiten}
\tilde{\bw}_k=\bC_\varepsilon^{1/2}\bw_k,\qquad \tilde{\bh}_k=\bC_\varepsilon^{-1/2}\bh_k,
\end{equation}
which leave all inner products invariant, $\bh_k^{\mathsf H}\bw_i=\tilde{\bh}_k^{\mathsf H}\tilde{\bw}_i$, and reduce $\mathrm{C}_1$ to the standard $\sum_k\|\tilde{\bw}_k\|^2\le P$. In whitened coordinates the block subproblem is the canonical MU-MISO weighted sum-rate maximization, which we solve by the weighted-MMSE (WMMSE) equivalence \cite{Shi2011}. Recall its logic, since we will reuse its pieces: for any receive scalar $u_k$ the mean-square error of stream $k$ is
\begin{equation}\label{eq:msedef}
e_k(u_k,\tilde{\mathbf{W}})=\big|1-u_k^{*}\tilde{\bh}_k^{\mathsf H}\tilde{\bw}_k\big|^2+|u_k|^2\Big(\sum_{i\ne k}\big|\tilde{\bh}_k^{\mathsf H}\tilde{\bw}_i\big|^2+\sigma_k^2\Big),
\end{equation}
minimized by the MMSE receiver and attaining $e_k^{\rm mmse}=(1+\Gamma_k)^{-1}$, so that maximizing $\sum_k\alpha_k\log_2(1+\Gamma_k)$ is equivalent to minimizing $\sum_k v_k e_k-\alpha_k\log_2 v_k$ jointly in $(\mathbf{u},\mathbf{v},\tilde{\mathbf{W}})$. Block-coordinate minimization gives the familiar updates, iterated until the objective stabilizes:
\begin{align}
u_k&=\Big(\textstyle\sum_{i}|\tilde{\bh}_k^{\mathsf H}\tilde{\bw}_i|^2+\sigma_k^2\Big)^{-1}\tilde{\bh}_k^{\mathsf H}\tilde{\bw}_k,\label{eq:wmmse-u}\\
v_k&=\alpha_k/e_k(u_k,\tilde{\mathbf{W}}),\label{eq:wmmse-v}\\
\tilde{\bw}_k&=v_k u_k\Big(\textstyle\sum_{i}v_i|u_i|^2\,\tilde{\bh}_i\tilde{\bh}_i^{\mathsf H}+\mu\,\bI_A\Big)^{-1}\tilde{\bh}_k,\label{eq:wmmse-w}
\end{align}
where \eqref{eq:wmmse-w} is the stationarity condition of the quadratic weighted-MSE Lagrangian and $\mu\ge0$ is found by bisection so that $\sum_k\|\tilde{\bw}_k\|^2=P$ (or $\mu=0$ if the constraint is inactive. The norm is monotonically decreasing in $\mu$, so bisection converges). Each pass is monotonically non-decreasing in the weighted sum rate \cite{Shi2011}. When $\mathrm{C}_2$ is enforced, the whitened subproblem remains convex, each voltage cap being a second-order-cone constraint in $\mathbf{W}$, and is solved by a conic solver or, preserving closed forms, by augmenting \eqref{eq:wmmse-w} with the penalty $\sum_a\nu_a|\mathbf{e}_a^{\mathsf T}\bar\bZ\bw_k|^2$ and updating the multipliers $\nu_a\ge0$ by projected subgradient. Precoders in current units are recovered as $\bw_k=\bC_\varepsilon^{-1/2}\tilde{\bw}_k$.

\begin{remark}[Superdirectivity margin]\label{rem:superdir}
Fix a user and normalize powers by $R_r/2$: $P_\Omega=\varepsilon\|\bw\|^2$, $P_{\rm rad}=\bw^{\mathsf H}\bC\bw$, $P_{\rm em}=\bw^{\mathsf H}\bC_\varepsilon\bw$, so the agnostic norm constraint prices ohmic dissipation, not radiation. Under $P_{\rm em}\le P$, beamforming attains $\max_{\bw}|\bh^{\mathsf H}\bw|^2=P\,\bh^{\mathsf H}\bC_\varepsilon^{-1}\bh$ and, on the eigenpairs $(\lambda_i,\mathbf{v}_i)$ of $\bC$ with weights $\beta_i=|\mathbf{v}_i^{\mathsf H}\bh|^2/\|\bh\|^2$,
\begin{align}
\frac{\|\bh\|^2 P}{\varepsilon+\lambda_{\max}(\bC)}\;&\le\;
\|\bh\|^2P\sum_i\frac{\beta_i}{\varepsilon+\lambda_i}\nonumber\\
&\le\;\frac{\|\bh\|^2 P}{\varepsilon+\lambda_{\min}(\bC)}\;\le\;\frac{\|\bh\|^2 P}{\varepsilon},\label{eq:sdchain}
\end{align}
against the uncoupled benchmark $\|\bh\|^2P$. The margin $\sum_i\beta_i/(\varepsilon+\lambda_i)$ exceeds one as soon as the channel places weight on eigenmodes with $\lambda_i<1-\varepsilon$, and it grows as those eigenvalues shrink, each term saturating at $\beta_i/\varepsilon$. Since $\lambda_{\min}(\bC)\to0$ at sub-wavelength spacings, gains beyond the port count follow, the mechanism of \cite{Liao2026MC,ZhangEMFAA2026}. The electrical size that governs this regime is the span of the active constellation rather than the hosting aperture: $\bC$ depends on the ports only through their pairwise separations, so a compact constellation on a large surface behaves as a small effective surface, and flexible positioning tunes this size directly, an observation quantified in Section~\ref{sec:numerics}. The chain also shows the caps: reaching $\|\bh\|^2P/\varepsilon$ requires $\|\bw\|^2\to P/\varepsilon$, hence vanishing efficiency $\lambda/(\varepsilon+\lambda)$ along the exploited mode (losses \cite{WangCAPA2025,Hansen1981}), while the inflated current norm raises the stored energy through $\Imag\{\bar\bZ\}$, policed by $\mathrm{C}_2$.
\end{remark}

\subsection{The position block: analytic gradients}\label{sec:posgrad}
Fix $\mathbf{W}$ and consider the positions. Introduce the partial Lagrangian
\begin{align}
\mathcal{L}(\Utt)\!&=\!\!\sum_{k}\!\alpha_k\log_2\!\big(1\!+\!\Gamma_k(\Utt,\mathbf{W})\big)\!-\!\mu\Big(\!\sum_k \bw_k^{\mathsf H}\bC_\varepsilon(\Utt)\bw_k\!-\!P\Big)\nonumber\\
&-\sum_a\nu_a\Big(\sum_k\big|\mathbf{e}_a^{\mathsf T}\bar\bZ(\Utt)\bw_k\big|^2\!-\!\bar V_{\max}^2\Big),\label{eq:lagr}
\end{align}
with $\mu$ and the multipliers $\nu_a\ge0$ inherited from the precoder block, $\nu_a=0$ whenever $\mathrm{C}_2$ is inactive. Every ingredient of the ascent is closed form, the practical dividend of the field-first model: no numerical differentiation of impedance matrices, no Sylvester machinery \cite{Liao2026MC}.

\begin{lemma}[Closed-form position gradients]\label{lem:grad}
Let $x=\kappa\tau$, $\hat{\boldsymbol{\tau}}=\bar{\boldsymbol{\tau}}/\tau$, $c=\bar{\bp}^{\mathsf T}\hat{\boldsymbol{\tau}}$, $\beta_k=\alpha_k/\ln2$, $S_k=|\bh_k^{\mathsf H}\bw_k|^2$, $I_k=\sum_{i\ne k}|\bh_k^{\mathsf H}\bw_i|^2+\sigma_k^2$, and $\mathbf{d}_{k,a}=\nabla_{\bu_a}[\bh_k]_a$. Then:
(i) the coupling gradient is
\begin{equation}\label{eq:gradrho}
\nabla\rho(\bar{\boldsymbol{\tau}})=\frac{3}{2}\Big[\big(f_0'(x)+j_2'(x)\,c^2\big)\,\kappa\,\hat{\boldsymbol{\tau}}
+\frac{2\,j_2(x)\,c}{\tau}\big(\bI_2-\hat{\boldsymbol{\tau}}\hat{\boldsymbol{\tau}}^{\mathsf T}\big)\bar{\bp}\Big],
\end{equation}
with $f_0'=-j_1-j_0/x+3j_1/x^2$ and $j_2'=j_1-3j_2/x$, and the substitution $j_\ell\mapsto h_\ell$ yields $\nabla\vartheta$. Consequently
\begin{equation}\label{eq:gradC}
\nabla_{\bu_a}\sum_k \bw_k^{\mathsf H}\bC_\varepsilon\bw_k
=2\sum_{k}\Real\Big\{w_{a,k}^{*}\sum_{b\ne a}w_{b,k}\,\nabla\rho(\bu_a-\bu_b)\Big\}.
\end{equation} and, with $v_{a,k}\triangleq\mathbf{e}_a^{\mathsf T}\bar\bZ(\Utt)\bw_k$ and $J_V\triangleq\sum_{a}\nu_{a}\sum_k|v_{a,k}|^2$ the voltage penalty of Section~\ref{sec:wmmse},
\begin{equation}\label{eq:gradV}
\nabla_{\bu_a}J_V\!=\!2\!\sum_k\!\Real\!\Big\{\!\sum_{b\ne a}\!\nabla\vartheta(\bu_a-\bu_b)\!\big(\nu_av_{a,k}^{*}w_{b,k}+\nu_bv_{b,k}^{*}w_{a,k}\big)\!\Big\},
\end{equation}
its second term recording that moving one port perturbs the drive voltage of every neighbor. Setting $\vartheta\mapsto\rho$, $\nu_a\mapsto1$, $v_{a,k}\mapsto w_{a,k}$ recovers \eqref{eq:gradC}.

(ii) the rate gradient is
\begin{equation}\label{eq:gradrate}
\nabla_{\bu_a}\!\sum_k\alpha_k\log_2(1+\Gamma_k)
=\sum_k \frac{2\beta_k}{I_k+S_k}\,
\Real\big\{\zeta_{k,a}\,\mathbf{d}_{k,a}^{*}\big\},
\end{equation}
with $\zeta_{k,a}=(\bw_k^{\mathsf H}\bh_k)\,w_{a,k}-\Gamma_k\sum_{i\ne k}(\bw_i^{\mathsf H}\bh_k)\,w_{a,i}$.

(iii) the channel gradient is, in general, $\mathbf{d}_{k,a}=-\int_\surf \mathbf{g}_k^{\mathsf H}(\bs)\bp\,\nabla\xi(\bar{\bs}-\bu_a)\,\dd\bar{\bs}$, evaluable for any profile through the spectral series of Section~\ref{sec:wnimpl}. In the small-element regime, with $\boldsymbol{\tau}=\br_k-[\bu_a^{\mathsf T},0]^{\mathsf T}$, $R=\|\boldsymbol{\tau}\|$, and $\boldsymbol{\Pi}=[\bI_2\ \mathbf{0}]\in\mathbb{R}^{2\times3}$,
\begin{align}
\mathbf{d}_{k,a}&=-\im\kappa Z_0\ell\,\boldsymbol{\Pi}\Big[\big(a'(R)\,\boldsymbol{\psi}_k^{\mathsf H}\bp+b'(R)\,(\boldsymbol{\psi}_k^{\mathsf H}\hat{\boldsymbol{\tau}})(\hat{\boldsymbol{\tau}}^{\mathsf T}\bp)\big)\hat{\boldsymbol{\tau}}\nonumber\\
&+\frac{b(R)}{R}\big(\bI_3-\hat{\boldsymbol{\tau}}\hat{\boldsymbol{\tau}}^{\mathsf T}\big)\big((\hat{\boldsymbol{\tau}}^{\mathsf T}\bp)\,\boldsymbol{\psi}_k^{*}+(\boldsymbol{\psi}_k^{\mathsf H}\hat{\boldsymbol{\tau}})\,\bp\big)\Big],\label{eq:gradh}
\end{align}
with $a'$, $b'$ reported in Appendix~\ref{app:grad}.
\end{lemma}
\begin{IEEEproof}
Please refer to Appendix~\ref{app:grad}.
\end{IEEEproof}

The update is $\bu_a^{+}=\Pi_{\surf}\big[\bu_a+\eta\,\nabla_{\bu_a}\mathcal{L}\big]$ for each port, with $\Pi_\surf$ the clip onto the aperture, the step $\eta$ by Armijo backtracking on the objective of $(\mathrm{P}1)$ after rescaling $\mathbf{W}$ to feasibility, and a pairwise push-out enforcing $\mathrm{C}_3$.

\subsection{Wavenumber-domain relaxation, upper bound, and initialization}\label{sec:wnimpl}
Theorem~\ref{thm:fas} is exploited algorithmically in three ways. First, as the data structure the algorithm runs on: writing $\boldsymbol{\Phi}(\Utt)\triangleq[\bm{\varphi}(\bu_1),\dots,\bm{\varphi}(\bu_A)]$ and stacking the tapered spectra $\bar{g}_{k,n}\triangleq\widehat{g}_{k,n}\,\widehat{\xi}(\bk_n)$ over the $N=O(\kappa^2|\surf|)$ propagating lattice modes, the channels of \eqref{eq:rxmu} and the gradients of \eqref{eq:gradrate} are codeword inner products,
\begin{equation}\label{eq:wnchan}
\bh_k(\Utt)=\boldsymbol{\Phi}(\Utt)^{\mathsf H}\bar{\mathbf{g}}_k,\qquad
\mathbf{d}_{k,a}=-\im\sum_{n}\bar{g}_{k,n}\,\bk_n\,[\bm{\varphi}(\bu_a)]_n,
\end{equation}
so the optimization consumes only the visible-mode coefficients per user. The definition of $\widehat{g}_{k,n}$ involves the user location through \eqref{eq:contchan}. The point is operational: once the coefficients are available, estimated from port measurements as in \cite{New2025Oversampling}, the iterations never evaluate the continuous channel again, and no geometry or propagation model enters the loop.

Second, as a computable performance certificate. Dropping the manifold constraint \eqref{eq:manifold}, i.e., letting the mode coefficients range over all of $\mathbb{C}^N$, defines the holographic relaxation
\begin{align}
(\mathrm{P}2):\ \max_{\tilde{\bq}_1,\dots,\tilde{\bq}_K\in\mathbb{C}^{N}}\ &\sum_{k}\alpha_k\log_2\Big(1+\frac{|\widehat{\mathbf{g}}_k^{\mathsf T}\tilde{\bq}_k|^2}{\sum_{i\ne k}|\widehat{\mathbf{g}}_k^{\mathsf T}\tilde{\bq}_i|^2+\sigma_k^2}\Big)\nonumber\\
\text{s.t.}\ &\ \sum_{k}\tilde{\bq}_k^{\mathsf H}\,\bar\bB_R\,\tilde{\bq}_k\le P,\label{eq:P2}
\end{align}
with $\widehat{\mathbf{g}}_k$ stacking the raw spectra $\widehat{g}_{k,n}$ and $\bar\bB_R\triangleq\bB_R/R_r$ the resistive Gram in the $R_r/2$ units of $\mathrm{C}_1$. Every port configuration synthesizes a current whose mode coefficients lie in $\mathcal{Q}_A\subset\mathbb{C}^N$, with the same received functional by Theorem~\ref{thm:fas}(i) and the same physical power by \eqref{eq:modalexact}. Hence the optimal value of $(\mathrm{P}2)$ upper-bounds that of $(\mathrm{P}1)$ for every $A$ and every port geometry, and its gap to a returned design certifies the price of $A$-port sparsity. Please note that $\mathrm{C}_2$ is dropped without harm to the bound, since removing constraints only enlarges the feasible set, and could not be carried into the mode domain regardless: port voltages exist only once ports do, and by Proposition~\ref{prop:split} a visible-only modal current has divergent reactive energy, the very reason reactive bookkeeping lives at port level. The constraint re-enters the moment the relaxed solution is projected onto ports. The relaxation is also where coupling costs least: by Lemma~\ref{lem:diag}, $\bar\bB_R$ is asymptotically $\diag\{\widehat{r}_\bp(\bk_n)/R_r\}$, so whitening reduces to an entrywise scaling of the $\tilde{q}_{n}$, after which $(\mathrm{P}2)$ is a standard MU-MISO weighted sum-rate problem solved by the WMMSE iterations \eqref{eq:wmmse-u}-\eqref{eq:wmmse-w} at $O(NK^2)$ per pass via the matrix-inversion lemma. For finite apertures the exact Gram \eqref{eq:Bexact} can be retained at a one-time $O(N^2)$ cost.

Third, as an initialization that removes the multistart burden of the position block, whose landscape oscillates on the wavelength scale. From the relaxed solution, ports are extracted by matching pursuit on the codeword dictionary: the correlation $(\widehat{\xi}\odot\bm{\varphi}(\bu))^{\mathsf H}\tilde{\bq}$, evaluated over a regular grid, is a single two-dimensional inverse FFT. Its peak fixes a port, the atom is deflated, and $A$ repetitions cost $O(A\,N\log N)$. Running the precoder block once at the extracted positions already returns a complete design, the inversion of $(\mathrm{P}2)$ onto the FAS manifold. The resulting schemes, together with the coupling-agnostic and half-wavelength baselines, are defined and compared in Section~\ref{sec:numerics}. When hardware confines the ports to a candidate grid of $N_g\gg A$ locations, the same FFT supplies all candidate channels at once, and a coupling-aware greedy selection costs $O(N_gA^2)$ per added port, marginal gains being rank-one updates of $\bC_\varepsilon^{-1}$ by Sherman-Morrison. The closed-form route of Section~\ref{sec:posgrad} and the spectral route \eqref{eq:wnchan} thus coexist: the former when a propagation model is trusted, the latter when only measured spectra are available.

\subsection{Algorithm, convergence, and complexity}
The overall procedure alternates the two blocks:
\begin{enumerate}
\item[1)] initialize $\Utt^{(0)}$ by the matching pursuit of Section~\ref{sec:wnimpl}, or by the $\lambda/2$ grid for which $\bC\approx\bI_A$ by Corollary~\ref{cor:halflambda} and $\mathbf{W}^{(0)}$ (whitened equal-power MRT).
\item[2)] \textbf{precoder block:} with $\Utt$ fixed, whiten via \eqref{eq:whiten} and iterate \eqref{eq:wmmse-u}-\eqref{eq:wmmse-w} to convergence.
\item[3)] \textbf{position block:} with $\mathbf{W}$ fixed, perform $T$ projected-gradient ascent steps on \eqref{eq:lagr} using \eqref{eq:gradC}-\eqref{eq:gradh}.
\item[4)] repeat 2)-3) until the gain falls below tolerance.
\end{enumerate}
Step 2) is monotone in the objective \cite{Shi2011}. Step 3) is monotone by the backtracking rule. Hence the sum-rate sequence is non-decreasing and, being bounded above (finite $P$, bounded channels), converges. Standard arguments then yield convergence to a partial stationary point of (P1). Per outer iteration: whitening costs one $A\times A$ eigendecomposition, $O(A^3)$. A WMMSE pass costs $O(K^2A)$ for the inner products in \eqref{eq:wmmse-u} and $O(KA^2+A^3)$ per bisection step for assembling and inverting the matrix in \eqref{eq:wmmse-w}. Each of the $T$ position steps costs $O(KA^2)$ for \eqref{eq:gradC} and $O(KA)$ for \eqref{eq:gradrate}-\eqref{eq:gradh}. All evaluations are closed form: no full-wave solver, no matrix-derivative subroutine, no Sylvester solves in the loop. On the spectral side, a WMMSE pass on $(\mathrm{P}2)$ costs $O(NK^2)$ and MP costs $O(A\,N\log N)$. End to end, the modal side is therefore independent of the port count, HUB costing $O(I_BNK^2)$ per design and MP adding only the extraction plus one precoder pass, while the port side totals $O(TI_B(KA^2+A^3))$ and overtakes the modal cost near $A\approx(NK^2/T)^{1/3}$, about five in the setup of Section~\ref{sec:numerics}: the wavenumber domain absorbs the aperture, the port domain absorbs the hardware, and the FFT bridge between them costs less than either.

\section{Numerical Results}\label{sec:numerics}

\subsection{Scenario, parameters, and baselines}
The transmit surface is $\surf=6\lambda\times6\lambda$, carrying $N=137$ visible lattice modes. $K=3$ single-antenna users with random polarizations are dropped uniformly over a $7.2\lambda\times7.2\lambda$ footprint at depths $z_k\in[6,14]\lambda$. The Rayleigh distance of the aperture is $2D^2/\lambda\approx144\lambda$, so every user lies deep in the Fresnel zone, the regime discussed after Lemma~\ref{lem:diag}: the spectra $\widehat{g}_{k,n}$ spread over many visible modes and no far-field assumption is invoked anywhere. Ports are Gaussian elements of effective support $\Delta\approx\lambda/5$ ($\sigma=\lambda/20$), for which $\widehat{\xi}^{\,2}\ge0.9$ over the visible disc, so the small-element closed forms of Lemmas~\ref{lem:closed} and~\ref{lem:grad} apply within a few percent, and the rim shading $\widehat{\xi}^{\,2}(\kappa)=e^{-(\kappa\sigma)^2}\approx0.91$ reappears below as the alias scale of Fig.~\ref{fig:alias}. Channels are generated by the exact dyadic Green function \eqref{eq:greenclosed} and projected on the lattice, and each user spectrum is normalized to unit average port gain, so the budget-to-noise ratio is the operating SNR. The loss ratio is $\varepsilon=0.05$, the polarization is $\bp=\hat{\mathbf{x}}$, and ports respect $d_{\min}=0.15\lambda$, living on a $48\times48$ candidate grid of spacing $\approx\lambda/8$ when discrete. The alternating optimization runs $10$ outer iterations of $30$ WMMSE passes and $4$ position steps, and every curve averages $6$ independent user drops.

Five schemes are compared. HUB, the holographic upper bound, is the WMMSE solution of the relaxation $(\mathrm{P}2)$ over all $N$ modal coefficients with the manifold constraint dropped, an upper bound for every port configuration. MP, the matching-pursuit design, is the inversion of $(\mathrm{P}2)$ onto the codebook: positions extracted by deflation on the $\lambda/8$ candidate grid followed by a single precoder pass, the port-selection (FAS) instance of Remark~\ref{rem:unif}. AO is the alternating optimization of Sections~\ref{sec:wmmse}-\ref{sec:posgrad} over continuous positions, the movable-antenna instance, initialized at MP. AGN is the coupling-agnostic design, the same alternating optimization run under the identity model $\bC_\varepsilon=\bI_A$ and then evaluated under the physical constraint \eqref{eq:powcon}, which quantifies what the classical norm model over-reports. MIMO is coupling-aware precoding on the frozen half-wavelength grid of Corollary~\ref{cor:halflambda}. Three gaps organize it: HUB to AO prices $A$-port sparsity, MP to AO prices refinement, and AO to AGN prices coupling awareness.

\begin{figure}[!t]
\centering
\includegraphics[width=.8\columnwidth]{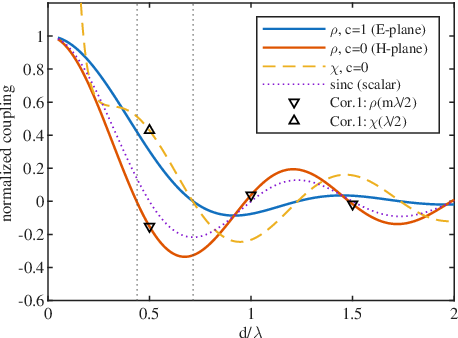}
\caption{Coupling kernels versus spacing with the nulls and half-wavelength residuals of Corollary~\ref{cor:halflambda} overlaid, and the scalar sinc for reference.}
\label{fig:kernels}
\end{figure}
\subsection{Convergence}
Fig.~\ref{fig:conv} shows the sum rate along the outer iterations for $A\in\{4,8\}$ and three initializations at $10$~dB, averaged over four drops. Every trace is monotone, as guaranteed by the block structure of Section~\ref{sec:opt}. The spectral initialization starts $1$ to $2$~bit above the half-wavelength start and reaches, in a single deterministic run, the level of the best random restart, confirming the role assigned to the relaxation in Section~\ref{sec:wnimpl}: the position landscape oscillates on the wavelength scale, and $(\mathrm{P}2)$ sees through it.

\begin{figure}[!t]
\centering
\includegraphics[width=.8\columnwidth]{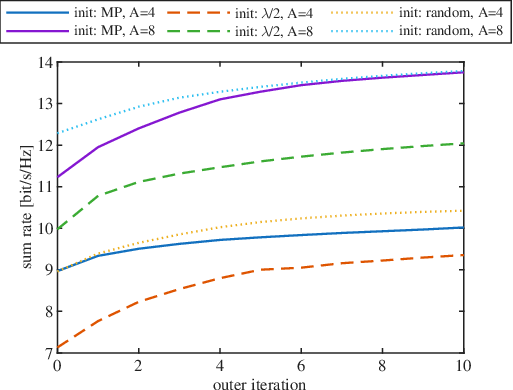}
\caption{Convergence of the alternating optimization for $A\in\{4,8\}$ under three initializations at $10$~dB: the spectral start matches the best random restart without multistart.}
\label{fig:conv}
\end{figure}

\subsection{Kernels, diagonalization, and the pullback}
Fig.~\ref{fig:kernels} plots the closed-form kernels of Lemma~\ref{lem:closed} against spacing. The resistive curves cross zero at $0.44\lambda$ and $0.72\lambda$, the roots of the null equations of Corollary~\ref{cor:halflambda} marked by the dotted lines, and pass exactly through the residual markers $3(-1)^m/2\pi^2m^2$ at $d=m\lambda/2$. The reactive kernel diverges as the spacing closes and still equals $0.43$ at $d=\lambda/2$, three times the resistive residual there: even the classical grid is reactively coupled, the observation that motivates carrying the voltage constraint \eqref{eq:voltcon} far from the superdirective regime.

Fig.~\ref{fig:diag} validates Lemma~\ref{lem:diag}. For one lattice mode pinned at each fractional radius $\|\bk_n\|/\kappa\in\{0.15,0.5,0.85\}$, the relative diagonalization error of $\bB_R$, diagonal deviation plus leakage to the nearest neighbors, decays monotonically with the aperture and straddles the $S^{-1/2}$ guide. The near-rim representative starts an order of magnitude above the interior ones, the numerical face of the H\"older rim term in the proof, and the interior modes are already within a few percent of their limits at $S=8\lambda$: the diagonal pricing that $(\mathrm{P}2)$ relies on is accurate at realistic apertures.

\begin{figure}[!t]
\centering
\includegraphics[width=.8\columnwidth]{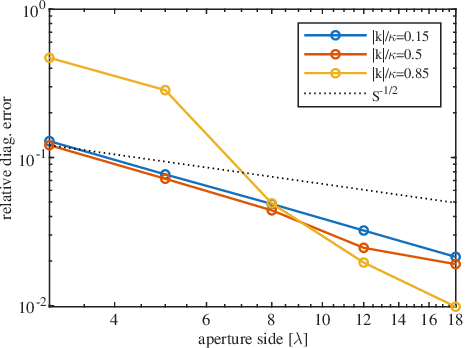}
\caption{Relative diagonalization error of $\bB_R$ for modes pinned at three radii, against the $S^{-1/2}$ guide of Lemma~\ref{lem:diag}: rim modes converge last.}
\label{fig:diag}
\end{figure}

\begin{figure}[!t]
\centering
\includegraphics[width=.8\columnwidth]{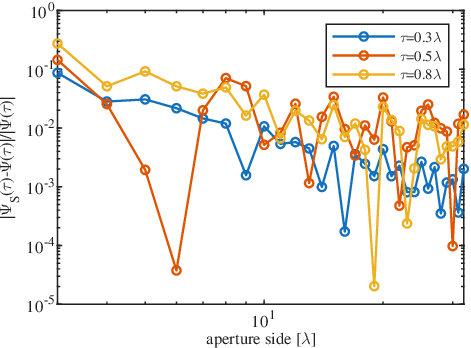}
\caption{Lattice pullback versus its limit \eqref{eq:pullback}: the error is the coupling to periodic images, a $1/S$ envelope with deterministic interference dips.}
\label{fig:alias}
\end{figure}

Fig.~\ref{fig:alias} examines the pullback of Theorem~\ref{thm:fas}(ii) at finite aperture. Let $\Psi_S$ denote the partial sum in \eqref{eq:pullback} at aperture side $S$, before the limit that defines $\Psi$. The relative error $|\Psi_S-\Psi|/|\Psi|$ decays with a $1/S$ envelope but is strongly non-monotone, with near-cancellations reaching $10^{-4}$. Both features are structural. By Poisson summation, sampling the spectrum on the lattice periodizes space, so $\Psi_S-\Psi$ is exactly the coupling to periodic images at distances $\|\bar{\boldsymbol{\tau}}+S\mathbf{m}\|$ through the kernel tail $\tfrac32(1-c^2)\sin(\kappa r)/\kappa r$ of Lemma~\ref{lem:closed}, scaled by the element factor $\widehat{\xi}^{\,2}(\kappa)\approx0.91$: the envelope is the tail decay, and a dip is an aperture at which the images interfere destructively, as at $\bar{\tau}=\lambda/2$ where the axial images sit on exact nulls. The factor is a knob: larger elements shade the rim and suppress every image at once, so the port size can be calibrated to bury the alias floor below other errors.

\begin{figure}[!t]
\centering
\includegraphics[width=.8\columnwidth]{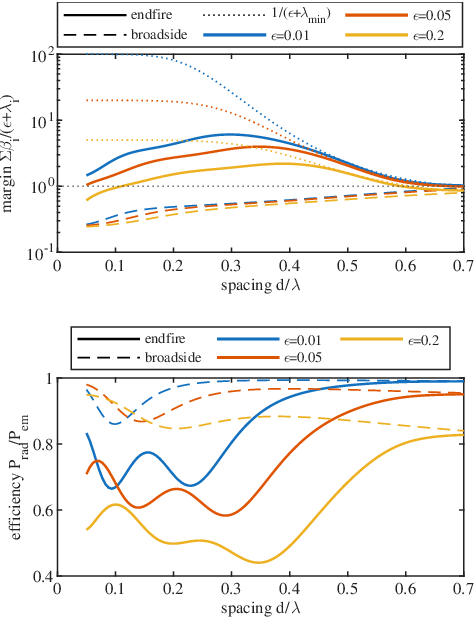}
\caption{Superdirectivity margin of \eqref{eq:sdchain} for endfire (solid) and broadside (dashed) steering with the envelope $1/(\varepsilon+\lambda_{\min})$ (dotted), and the radiation efficiency beneath.}
\label{fig:margin}
\end{figure}

\subsection{The superdirectivity margin}
Fig.~\ref{fig:margin} evaluates Remark~\ref{rem:superdir} on a four-port row. For endfire steering the margin exceeds one below $d\approx0.55\lambda$ and peaks near $d\approx0.3\lambda$, approaching the best-case envelope $1/(\varepsilon+\lambda_{\min})$ that saturates at $1/\varepsilon$, while broadside steering stays below one at all spacings despite its near-unit efficiency: the two channels load opposite ends of the eigenvalue spectrum in \eqref{eq:sdchain}. The efficiency panel shows the price, radiated fractions dropping to $0.45$ to $0.7$ exactly where the margin peaks. The basis-free identity behind the figure is that the margin factors as efficiency times the directivity enhancement per radiated watt: at $\varepsilon=0.05$ and $d=0.3\lambda$, a margin of $3.8$ with efficiency $0.6$ corresponds to a pattern delivering $6.3$ times more power per radiated watt toward the endfire user than the uncoupled benchmark. Gains beyond the port count are therefore real, bought with current, and capped by the losses, the accounting that constraint $\mathrm{C}_2$ makes explicit. A geometric reading sharpens the mechanism. Since $\bC$ depends on the ports only through their separations, the operative length is the constellation span $L=(A-1)d$ rather than the hosting aperture: the synthesized current is windowed to the ports' reach, its spectrum resolves no finer than $2\pi/L$, and a far-field direction loads invisible-side eigenmodes only within the grazing sector $1-\sin\theta\lesssim\lambda/L$. Flexible positioning tunes the effective surface, contracting it to widen the superdirective sector or spreading it to restore the classical regime.

\subsection{Multi-user performance}
Fig.~\ref{fig:triplet} compares HUB, MP, and AO versus the number of ports at $5$ and $15$~dB. HUB is flat by construction, its $N$ modal degrees of freedom being independent of $A$, and quantifies the price of sparsity: about $19$~bit at $A=2$ and still $13$ to $15$~bit at $A=8$ at the higher SNR. The refinement gap is small throughout, MP landing within $1$ to $2$~bit of AO at every operating point: once the relaxed optimum is known, its FFT projection is nearly optimal, and the iterative machinery buys polish rather than substance. For port-selection hardware this is the practical message of Theorem~\ref{thm:fas}: solve the modal problem once, project, and precode.

\begin{figure}[!t]
\centering
\includegraphics[width=.8\columnwidth]{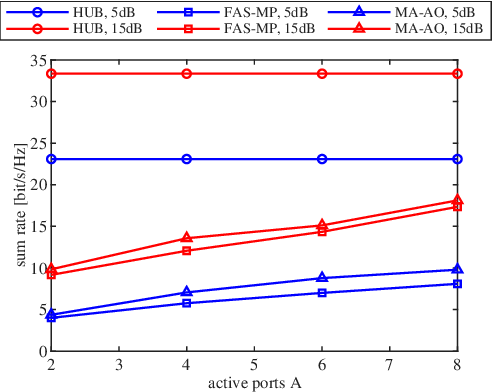}
\caption{HUB, MP, and AO versus the number of ports at $5$ and $15$~dB: the large HUB gap prices sparsity, the small MP gap prices refinement.}
\label{fig:triplet}
\end{figure}

\begin{figure}[!t]
\centering
\includegraphics[width=.8\columnwidth]{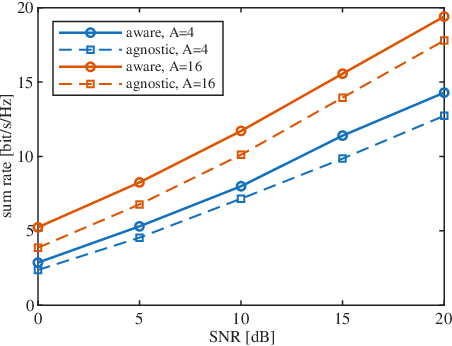}
\caption{Coupling-aware versus agnostic optimization on one placement pad: awareness is free at $A=4$ and decisive under the forced packing of $A=16$.}
\label{fig:aware}
\end{figure}

Fig.~\ref{fig:aware} isolates the value of coupling awareness. All curves share one placement pad, and only the port count changes: with $A=4$ the optimizer keeps spacings around $\lambda$, where Corollary~\ref{cor:halflambda} predicts near-vanishing resistive coupling, and indeed the aware and agnostic designs coincide. With $A=16$ the same pad forces spacings near $0.35\lambda$, where Fig.~\ref{fig:kernels} shows both kernels large, and the aware design opens a gap over the agnostic one that widens with SNR. The two regimes in one figure state the engineering conclusion of the paper: on an unconstrained aperture, optimized layouts self-decouple and the classical model survives, while whenever form factor or port density forces sub-half-wavelength packing, the impedance kernel must enter the design loop, and Sections~\ref{sec:wmmse}-\ref{sec:posgrad} show that it can do so entirely in closed form.

\section{Conclusion and Future Works}\label{sec:conclusion}
This paper placed the flexible-position antenna where it electromagnetically lives, on a holographic surface. A field-first multi-user model carries the full impedance kernel in its constraints, is provably equivalent to the multiport circuit models of the FAS/MA literature, and obeys an exact wavenumber dichotomy: the light circle separates the resistive spectrum, which prices average power, from the reactive spectrum, which prices drive voltages. In this representation the antenna is a constant-modulus wavenumber codeword, positioning is modulation and selection is codebook restriction, the pulled-back coupling reproduces the polarized nulls of continuous-aperture theory and collapses to uncoupled MIMO at half wavelength up to a quantified, predominantly reactive residual, and the coupling-aware sum-rate problem admits an alternating solution with closed-form gradients throughout. Future work includes matching-network co-design over the reactive part, the exact discrete-codebook variant of Section~\ref{sec:wnimpl}, coupling-aware channel estimation extending \cite{New2025Oversampling}, and propagation kernels beyond free space.

\appendices

\section{Proof of Lemma~\ref{lem:closed}}\label{app:closed}
With $g(R)=e^{\im\kappa R}/(4\pi R)$, $\nabla R=\hat{\boldsymbol{\tau}}$, and $\nabla\hat{\boldsymbol{\tau}}=(\bI_3-\hat{\boldsymbol{\tau}}\hat{\boldsymbol{\tau}}^{\mathsf T})/R$, one has $\nabla g=g'\hat{\boldsymbol{\tau}}$ with $g'=(\im\kappa-1/R)g$ and, by the product rule,
\begin{equation*}
\nabla\nabla^{\mathsf T}g=g''\,\hat{\boldsymbol{\tau}}\hat{\boldsymbol{\tau}}^{\mathsf T}+\frac{g'}{R}\big(\bI_3-\hat{\boldsymbol{\tau}}\hat{\boldsymbol{\tau}}^{\mathsf T}\big),\
g''=\Big(\!-\kappa^2-\frac{2\im\kappa}{R}+\frac{2}{R^2}\Big)g,
\end{equation*}
where $g''$ follows by differentiating $g'=(\im\kappa-1/R)g$ once more. Substituting into \eqref{eq:portpower},
\begin{equation*}
\bG=\Big(g+\frac{g'}{\kappa^2R}\Big)\bI_3+\frac{g''-g'/R}{\kappa^2}\,\hat{\boldsymbol{\tau}}\hat{\boldsymbol{\tau}}^{\mathsf T},
\end{equation*}
and collecting powers of $(\kappa R)^{-1}$ gives $a$, $b$ of \eqref{eq:hankelform}. For the Hankel form, with $h_0(x)=-\im e^{\im x}/x$ and $h_1(x)=-e^{\im x}(x+\im)/x^2$,
\begin{equation*}
\frac{\im\kappa}{4\pi}\Big(h_0-\frac{h_1}{x}\Big)=\frac{\kappa}{4\pi}\,e^{\im x}\Big(\frac1x+\frac{\im}{x^2}-\frac{1}{x^3}\Big)=a(R),
\end{equation*}
and, by the recurrence $h_2=(3/x)h_1-h_0=e^{\im x}(\im/x-3/x^2-3\im/x^3)$,
\begin{equation*}
\frac{\im\kappa}{4\pi}\,h_2=-\frac{\kappa}{4\pi}\,e^{\im x}\Big(\frac1x+\frac{3\im}{x^2}-\frac{3}{x^3}\Big)=b(R),
\end{equation*}
proving \eqref{eq:greenclosed}. The split \eqref{eq:imG}-\eqref{eq:reG} is immediate from $\im h_\ell=\im j_\ell-y_\ell$. Finally, each entry of \eqref{eq:imG} is an even entire series in $x$ with $j_0(0)=1$, $j_1(x)/x\to1/3$, $j_2(0)=0$, whence $\Imag\{\bG(\mathbf{0})\}=\tfrac{\kappa}{4\pi}\cdot\tfrac23\bI_3=\tfrac{\kappa}{6\pi}\bI_3$. The stated divergences of \eqref{eq:reG} follow from $y_0\sim-1/x$, $y_1\sim-1/x^2$, $y_2\sim-3/x^3$. $\hfill\blacksquare$

\section{Proof of Proposition~\ref{prop:power}}\label{app:poynting}
Step 1 (plane-wave decomposition \eqref{eq:planewave}). Aligning the polar axis with $\hat{\boldsymbol{\tau}}$ and substituting $u=\cos\theta$ in $\dd\Omega=\sin\theta\,\dd\theta\,\dd\phi$,
\begin{equation*}
\int_{\mathbb{S}^2}e^{\im x\,\hat{\mathbf{k}}\cdot\hat{\boldsymbol{\tau}}}\dd\Omega=2\pi\int_{-1}^{1}e^{\im xu}\dd u=4\pi\,j_0(x),
\end{equation*}
so $\Imag\{g(R)\}=\tfrac{\kappa}{4\pi}j_0(\kappa R)=\tfrac{\kappa}{16\pi^2}\int_{\mathbb{S}^2}e^{\im\kappa\hat{\mathbf{k}}\cdot\boldsymbol{\tau}}\dd\Omega$. Applying $(\bI_3+\kappa^{-2}\nabla\nabla^{\mathsf T})$ under the integral sign, where each $\nabla$ pulls down $\im\kappa\hat{\mathbf{k}}$ and produces the transverse projector $\bI_3-\hat{\mathbf{k}}\hat{\mathbf{k}}^{\mathsf T}$, gives the right-hand side of \eqref{eq:planewave}. That it equals $\Imag\{\bG\}$ is verified against Lemma~\ref{lem:closed} by computing the two angular moments. First, the scalar moment above. Second, by rotational symmetry about $\hat{\boldsymbol{\tau}}$, $\mathbf{T}\triangleq\int_{\mathbb{S}^2}\hat{\mathbf{k}}\hat{\mathbf{k}}^{\mathsf T}e^{\im x\hat{\mathbf{k}}\cdot\hat{\boldsymbol{\tau}}}\dd\Omega=\alpha\bI_3+\beta\hat{\boldsymbol{\tau}}\hat{\boldsymbol{\tau}}^{\mathsf T}$. The trace ($\operatorname{tr}\{\hat{\mathbf{k}}\hat{\mathbf{k}}^{\mathsf T}\}=1$) gives $3\alpha+\beta=4\pi j_0$, while projecting on $\hat{\boldsymbol{\tau}}$ and differentiating $\int_{-1}^1e^{\im xu}\dd u=2j_0$ twice gives $\alpha+\beta=2\pi\int_{-1}^1u^2e^{\im xu}\dd u=-4\pi j_0''=4\pi(j_0-2j_1/x)$, using $j_0''=-j_0+2j_1/x$ from the spherical Bessel equation. Solving, $\alpha=4\pi j_1/x$ and $\beta=-4\pi j_2$ (recurrence $j_0+j_2=3j_1/x$), whence
\begin{equation*}
\frac{\kappa}{16\pi^2}\big[4\pi j_0\,\bI_3-\mathbf{T}\big]=\frac{\kappa}{4\pi}\Big[\Big(j_0-\frac{j_1}{x}\Big)\bI_3+j_2\hat{\boldsymbol{\tau}}\hat{\boldsymbol{\tau}}^{\mathsf T}\Big]=\Imag\{\bG\}. 
\end{equation*}

Step 2 (Poynting identity). In phasor form under $e^{-\im2\pi ft}$, the free-space Maxwell equations with the impressed sheet current read $\nabla\times\boldsymbol{\mathcal{E}}=\im\omega\mu_0\boldsymbol{\mathcal{H}}$ and $\nabla\times\boldsymbol{\mathcal{H}}=\mathbf{J}-\im\omega\varepsilon_0\boldsymbol{\mathcal{E}}$, with $\mathbf{J}(\br)=\bj(\bar{\br})\delta(z)$. For any two fields, expanding in components with the Levi-Civita symbol, $\nabla\cdot(\mathbf{a}\times\mathbf{b})=\mathbf{b}\cdot(\nabla\times\mathbf{a})-\mathbf{a}\cdot(\nabla\times\mathbf{b})$. Choosing $\mathbf{a}=\boldsymbol{\mathcal{E}}$, $\mathbf{b}=\boldsymbol{\mathcal{H}}^{*}$ and substituting the first equation and the conjugate of the second,
\begin{equation*}
\nabla\cdot(\boldsymbol{\mathcal{E}}\times\boldsymbol{\mathcal{H}}^{*})
=\im\omega\mu_0\|\boldsymbol{\mathcal{H}}\|^2-\im\omega\varepsilon_0\|\boldsymbol{\mathcal{E}}\|^2-\boldsymbol{\mathcal{E}}\cdot\mathbf{J}^{*}.
\end{equation*}
Solving for the source term, multiplying by $\tfrac12$, integrating over the ball $V_{r_0}\supset\surf$, and applying the divergence theorem (the sifting property of $\delta(z)$ reduces the source integral to the aperture),
\begin{equation}\label{eq:appPoy}
-\frac12\int_\surf \bj^{\mathsf H}\boldsymbol{\mathcal{E}}\,\dd\bs
=\frac12\oint_{r=r_0}\!\!\!\!\!\!\!\!(\boldsymbol{\mathcal{E}}\times\boldsymbol{\mathcal{H}}^{*})\cdot\hat{\mathbf{n}}\,\dd S
+\im\,2\omega\big(W_e(r_0)-W_m(r_0)\big),
\end{equation}
with $W_e(r_0)=\tfrac{\varepsilon_0}{4}\int_{V_{r_0}}\|\boldsymbol{\mathcal{E}}\|^2$ and $W_m(r_0)=\tfrac{\mu_0}{4}\int_{V_{r_0}}\|\boldsymbol{\mathcal{H}}\|^2$ (the $1/4$ because time-averaging a squared cosine halves the instantaneous $\varepsilon_0/2$, $\mu_0/2$ densities).

Step 3 (far-field limit and the flux). Identity \eqref{eq:appPoy} holds for every $r_0$, with an $r_0$-independent left side. At finite radius, however, the flux is complex, carrying the escaping power together with reactive circulation across the sphere, so the split of the right side into radiated power and stored energy is not yet meaningful. The limit $r_0\to\infty$ is what disentangles the two, because reactive circulation is a near-field phenomenon: for $r\to\infty$ with $\bs$ over the bounded $\surf$, expanding $\|\br-\bs\|=r-\hat{\mathbf{r}}\cdot\bs+O(r^{-1})$ in the phase, $\|\br-\bs\|^{-1}=r^{-1}+O(r^{-2})$ in the amplitude, and using the projector limit of Lemma~\ref{lem:closed},
\begin{align}
\boldsymbol{\mathcal{E}}(\br)&=\mathbf{f}(\hat{\mathbf{r}})\,\frac{e^{\im\kappa r}}{r}+O(r^{-2}),\label{eq:farfield}\\
\mathbf{f}(\hat{\mathbf{r}})&=\frac{\im\kappa Z_0}{4\pi}\big(\bI_3-\hat{\mathbf{r}}\hat{\mathbf{r}}^{\mathsf T}\big)\tilde{\bj}(\kappa\hat{\mathbf{r}}),\nonumber
\end{align}
with $\boldsymbol{\mathcal{H}}=\hat{\mathbf{r}}\times\boldsymbol{\mathcal{E}}/Z_0+O(r^{-2})$ and $\hat{\mathbf{r}}^{\mathsf T}\mathbf{f}=0$. By $\mathbf{a}\times(\mathbf{b}\times\mathbf{c})=\mathbf{b}(\mathbf{a}\cdot\mathbf{c})-\mathbf{c}(\mathbf{a}\cdot\mathbf{b})$ and transversality,
$\boldsymbol{\mathcal{E}}\times\boldsymbol{\mathcal{H}}^{*}=\|\mathbf{f}\|^2\hat{\mathbf{r}}/(Z_0r^2)+O(r^{-3})$: real and outward at leading order. Since $\dd S=r_0^2\dd\Omega$, the flux converges to the real limit
\begin{align}
&\lim_{r_0\to\infty}\frac12\oint(\boldsymbol{\mathcal{E}}\times\boldsymbol{\mathcal{H}}^{*})\cdot\hat{\mathbf{n}}\,\dd S
=\frac{1}{2Z_0}\int_{\mathbb{S}^2}\|\mathbf{f}(\hat{\mathbf{k}})\|^2\dd\Omega\label{eq:fluxlimit}\\
&=\frac{\kappa^2Z_0}{32\pi^2}\int_{\mathbb{S}^2}\big\|\big(\bI_3-\hat{\mathbf{k}}\hat{\mathbf{k}}^{\mathsf T}\big)\tilde{\bj}(\kappa\hat{\mathbf{k}})\big\|^2\dd\Omega,\nonumber
\end{align}
its imaginary part, fed by the $O(r^{-3})$ cross terms against the $O(r_0^2)$ area, vanishing as $O(r_0^{-1})$. The same expansion gives $\mu_0\|\boldsymbol{\mathcal{H}}\|^2=\varepsilon_0\|\boldsymbol{\mathcal{E}}\|^2+O(r^{-3})$, so the individually divergent far-zone contributions to $W_e$ and $W_m$ cancel in the difference, which converges.

Step 4 (the flux is the kernel form). Substituting \eqref{eq:planewave} into the radiative quadratic form and exchanging the absolutely convergent integrals, the exponential of the difference $\bs-\bs'$ factorizes. The projector is constant in $(\bs,\bs')$, so the double aperture integral separates into $\tilde{\bj}(\kappa\hat{\mathbf{k}})^{\mathsf H}(\bI_3-\hat{\mathbf{k}}\hat{\mathbf{k}}^{\mathsf T})\tilde{\bj}(\kappa\hat{\mathbf{k}})$, the conjugation on $\bj^{\mathsf H}(\bs)$ supplying the sign flip that turns the $\bs$-integral into $\tilde{\bj}^{\mathsf H}$, and idempotency of the Hermitian projector turns the sandwich into $\|(\bI_3-\hat{\mathbf{k}}\hat{\mathbf{k}}^{\mathsf T})\tilde{\bj}(\kappa\hat{\mathbf{k}})\|^2$. Hence $\tfrac{\kappa Z_0}{2}\langle\bj,\Imag\{\bG\}\bj\rangle$ equals the right-hand side of \eqref{eq:fluxlimit}. Recalling from Section~II that both split forms are real, matching the real and imaginary parts of \eqref{eq:appPoy} against $-\tfrac12\int\bj^{\mathsf H}\boldsymbol{\mathcal{E}}=\tfrac{\kappa Z_0}{2}\langle\bj,\Imag\{\bG\}\bj\rangle+\im\tfrac12\langle\bj,\kerX\bj\rangle$ proves the three-way equality \eqref{eq:Pradforms} and item (iii) at once.

Step 5 (ohmic part: sheet limit). Let the conductor occupy $\surf\times[0,t]$ with conductivity $\sigma$ and $t$ much smaller than $\lambda$ and the skin depth, so the volumetric current density $\mathbf{J}_v$ is uniform across the thickness and $\bj=\mathbf{J}_v t$. Ohm's law $\mathbf{J}_v=\sigma\boldsymbol{\mathcal{E}}_{\rm in}$ fixes the internal field the generators must impress,
$\boldsymbol{\mathcal{E}}_{\rm in}=\bj/(\sigma t)=Z_s\bj$, distinct from the radiated $\boldsymbol{\mathcal{E}}$ of Steps 2-4. The Joule dissipation collapses to the sheet,
\begin{equation*}
\frac12\int_{\surf\times[0,t]}\!\!\!\!\!\!\!\!\!\!\!\!\sigma\|\boldsymbol{\mathcal{E}}_{\rm in}\|^2\dd V
=\frac{\sigma t}{2}\int_\surf\|Z_s\bj\|^2\dd\bs
=\frac{Z_s}{2}\int_\surf\|\bj\|^2\dd\bs=P_\Omega,
\end{equation*}
using $\sigma t Z_s^2=Z_s$, and the generator supply $\tfrac12\Real\int\bj^{\mathsf H}\boldsymbol{\mathcal{E}}_{\rm in}=P_\Omega$ is the first term of \eqref{eq:cplxpower}. Summing with Steps 2-4 yields the balance of the proposition, proving (ii) and completing the proof. $\hfill\blacksquare$

\section{Proof of Theorem~\ref{thm:equiv}}\label{app:equiv}
(i) Substituting $\bj=\sum_a\iota_a\boldsymbol{\xi}_a$ into \eqref{eq:cplxpower} and exchanging the finite sum with the integrals gives $P_{\rm cx}=\tfrac12\sum_{a,b}\iota_a^{*}\iota_b\langle\boldsymbol{\xi}_a,\kerZ\boldsymbol{\xi}_b\rangle$. Substituting $\bar{\bs}=\bu_a+\boldsymbol{\alpha}$, $\bar{\bs}'=\bu_b+\boldsymbol{\beta}$ and using that $\kerZ$ depends only on $\bar{\bs}-\bar{\bs}'=(\bu_a-\bu_b)+(\boldsymbol{\alpha}-\boldsymbol{\beta})$ yields the correlation form \eqref{eq:portZ}, the conjugate on $\xi$ dropping because the profile is real, the step that makes $\bZ$ complex symmetric, as network reciprocity demands. Uniqueness: if two complex symmetric matrices induce the same complex power for every $\boldsymbol{\iota}$, the real and imaginary parts of $\boldsymbol{\iota}^{\mathsf H}\bZ\boldsymbol{\iota}$ agree separately. Each is a Hermitian form in the real symmetric matrices $\Real\{\bZ\}$, $\Imag\{\bZ\}$, which polarization determines. Entrywise, the port voltage implied by \eqref{eq:cplxpower} is $v_a=\langle\boldsymbol{\xi}_a,\kerZ\bj\rangle$, so $[\bZ]_{ab}$ coincides with the induced-EMF mutual impedance $-\tfrac{1}{\iota_a\iota_b}\int\boldsymbol{\mathcal{E}}_b\cdot\bj_a$ of circuit theory \cite{Ivrlac2010}, term by term.

(ii) The profiles are real, so taking real and imaginary parts of \eqref{eq:portZ} commutes with the (real) integrations: $\Real\{[\bZ]_{ab}\}=\langle\boldsymbol{\xi}_a,\kerR\boldsymbol{\xi}_b\rangle$ and $\Imag\{[\bZ]_{ab}\}=\langle\boldsymbol{\xi}_a,\kerX\boldsymbol{\xi}_b\rangle$ by \eqref{eq:RXsplit}. The power identifications are Proposition~\ref{prop:power} applied to $\bj=\sum_a\iota_a\boldsymbol{\xi}_a$.

(iii) For $a\ne b$ the kernel $\zeta_\bp$ is smooth on the integration region as soon as the supports are disjoint. Expanding around $\bu_a-\bu_b$, the zeroth order contributes $(\int\xi)^2\zeta_\bp(\bu_a-\bu_b)=\ell^2\zeta_\bp(\bu_a-\bu_b)$. The first order vanishes by evenness of $\xi$. The second is $O((\kappa\Delta)^2)$ by smoothness. Evaluating via \eqref{eq:greenclosed}: $-\im\kappa Z_0\ell^2\,\bp^{\mathsf H}\bG\bp=-\im\kappa Z_0\ell^2\cdot\tfrac{\im\kappa}{4\pi}[h_0-\tfrac{h_1}{x}+h_2c^2]=R_r\vartheta$, the normalization $3/2$ in \eqref{eq:vartheta} fixed by $\Real\{\vartheta\}\to1$ as $x\to0$, and the identification of $\rho$, $\chi$ with \eqref{eq:imG}, \eqref{eq:reG} immediate from $h_\ell=j_\ell+\im y_\ell$. For the radiative self-term, $\Imag\{\bG\}$ is entire, so dominated convergence in the correlation gives $\kappa Z_0\ell^2\,\bp^{\mathsf H}\Imag\{\bG(\mathbf{0})\}\bp+O((\kappa\Delta)^2)=\kappa Z_0\ell^2\cdot\tfrac{\kappa}{6\pi}=R_r$, the textbook Hertzian radiation resistance, a consistency check of all constants.

(iv) For the reactive self-term the same argument fails, and must fail: by \eqref{eq:reG}, $\bp^{\mathsf H}\Real\{\bG(\bar{\boldsymbol{\tau}})\}\bp$ diverges as $\|\bar{\boldsymbol{\tau}}\|^{-3}$, so the correlation $(\xi\star\Real\{\bG\}\star\xi)(\mathbf{0})$ is finite for every profile with $\Delta>0$ but scales as $(\kappa\Delta)^{-3}$. Likewise $Z_s\int|\xi|^2\ge Z_s\ell^2/(\pi\Delta^2/4)$ by Cauchy-Schwarz. No universal point limit exists, and the values of $X_A$ and $R_\Omega$ are properties of the element, the information circuit models supply as element-specific constants. $\hfill\blacksquare$

\section{Proof of Proposition~\ref{prop:split}}\label{app:spectral}
The Weyl identity \cite[Ch.~2]{Chew1995} expands the scalar spherical wave into plane waves over the aperture plane: for $z\ne 0$,
\begin{equation}\label{eq:weyl}
\frac{e^{\im\kappa\|\boldsymbol{\tau}\|}}{4\pi\|\boldsymbol{\tau}\|}=\frac{\im}{8\pi^2}\int_{\mathbb{R}^2}\frac{e^{\im(\bk\cdot\bar{\boldsymbol{\tau}}+k_z|z|)}}{k_z}\,\dd\bk,\qquad k_z=\sqrt{\kappa^2-\|\bk\|^2},
\end{equation}
with $\Imag\{k_z\}\ge0$ so that evanescent components decay away from the source plane. Applying $(\bI_3+\kappa^{-2}\nabla\nabla^{\mathsf T})$ under the integral sign turns each plane wave into $(\bI_3-\mathbf{k}_{\pm}\mathbf{k}_{\pm}^{\mathsf T}/\kappa^2)e^{\im(\bk\cdot\bar{\boldsymbol{\tau}}+k_z|z|)}$ with $\mathbf{k}_\pm=[\bk^{\mathsf T},\pm k_z]^{\mathsf T}$, the sign following that of $z$. For source and observation both on the plane we take the transverse-transverse block and the limit $z\to0^{\pm}$, which is unambiguous for that block because, for tangential~$\bp$,
\begin{equation*}
\bp^{\mathsf H}\Big(\bI_3-\frac{\mathbf{k}_\pm\mathbf{k}_\pm^{\mathsf T}}{\kappa^2}\Big)\bp=1-\frac{(\bk\cdot\bar{\bp})^2}{\kappa^2}
\end{equation*}
is independent of the hemisphere and of $k_z$ (the discontinuous $z$-mixed blocks do not enter).\footnote{The exchange of $(\bI_3+\kappa^{-2}\nabla\nabla^{\mathsf T})$ with the $\bk$-integral and the limit $z\to0$ are the steps that require the distributional reading announced in the proposition: the resulting planar transform is a tempered distribution whose growing part encodes the $R^{-3}$ singularity of $\kerX$, exactly as the Fourier transform of $1/(4\pi\|\bar{\boldsymbol{\tau}}\|)$ is $1/(2\|\bk\|)$. All quadratic forms in which we use \eqref{eq:zspec} are evaluated against profiles or Dirichlet kernels whose decay makes them well defined, and this is tracked explicitly at each use.} Hence the planar transform of $-\im\kappa Z_0\,\bp^{\mathsf H}\bG\bp$ is
\begin{equation*}
-\im\kappa Z_0\cdot\im\,\frac{1-(\bk\cdot\bar{\bp})^2/\kappa^2}{2k_z}
=\kappa Z_0\,\frac{1-(\bk\cdot\bar{\bp})^2/\kappa^2}{2k_z},
\end{equation*}
and adding the flat transform $Z_s$ of the ohmic delta gives \eqref{eq:zspec}. The dichotomy \eqref{eq:rspec}-\eqref{eq:xspec} is now a pointwise statement: on the disc $k_z>0$ is real, so the radiation term is real (resistive). Outside, $k_z=\im|k_z|$ turns it into $-\im\kappa Z_0(1-(\bk\cdot\bar{\bp})^2/\kappa^2)/(2|k_z|)$, purely imaginary (reactive), with sign flipping across $|\bk\cdot\bar{\bp}|=\kappa$.

For the consistency check announced in the proof sketch, invert the disc part. In polar coordinates $\bk=\varrho[\cos\varphi,\sin\varphi]^{\mathsf T}$ with $\bar{\bp}=[1,0]^{\mathsf T}$ without loss of generality,
\begin{equation*}
\frac{1}{4\pi^2}\int_{\varrho<\kappa}\frac{1-\varrho^2\cos^2\varphi/\kappa^2}{2\sqrt{\kappa^2-\varrho^2}}\,e^{\im\bk\cdot\bar{\boldsymbol{\tau}}}\varrho\,\dd\varrho\,\dd\varphi ,
\end{equation*}
and at $\bar{\boldsymbol{\tau}}=\mathbf{0}$ the angular average of $\cos^2\varphi$ is $\tfrac12$, giving $\tfrac{1}{4\pi^2}\big[\tfrac12\cdot2\pi\kappa-\tfrac{1}{2\kappa^2}\cdot\pi\cdot\tfrac{2\kappa^3}{3}\big]=\tfrac{\kappa}{6\pi}$, where $\int_0^\kappa\varrho\,\dd\varrho/\sqrt{\kappa^2-\varrho^2}=\kappa$ and $\int_0^\kappa\varrho^3\dd\varrho/\sqrt{\kappa^2-\varrho^2}=\tfrac{2\kappa^3}{3}$: precisely $\bp^{\mathsf H}\Imag\{\bG(\mathbf{0})\}\bp$. The full inversion at $\bar{\boldsymbol{\tau}}\ne\mathbf{0}$ proceeds by the same moments combined with $\int_{\varrho<\kappa}e^{\im\bk\cdot\bar{\boldsymbol{\tau}}}k_z^{-1}\dd\bk=2\pi\sin(\kappa\tau)/\tau$ and its derivatives, and returns \eqref{eq:imG} restricted to in-plane separations, as claimed. $\hfill\blacksquare$

\section{Proof of Lemma~\ref{lem:diag}}\label{app:lemdiag}
Throughout, $\widehat{r}$ abbreviates the disc weight $\widehat{r}^{\,\rm rad}_\bp$ of \eqref{eq:Bexact}, extended by zero outside the disc.

Step 1 (lattice orthogonality and error form). $D(\bk-\bk_n)$ is the transform of the modulated indicator $\mathsf{1}_\surf(\bar{\bs})\,e^{\im\bk_n\cdot\bar{\bs}}$, so Parseval gives
\begin{equation*}
\frac{1}{4\pi^2}\int_{\mathbb{R}^2}\!\!\!\!D(\bk-\bk_n)D(\bk-\bk_m)\,\dd\bk=\!\int_{\surf}e^{\im(\bk_m-\bk_n)\cdot\bar{\bs}}\dd\bar{\bs}=|\surf|\,\delta_{nm},
\end{equation*}
exactly, lattice differences completing full periods over $\surf$. Subtracting $\widehat{r}(\bk_n)$ times this identity from \eqref{eq:Bexact},
\begin{equation}\label{eq:errsplit}
\varepsilon_{nm}=\frac{1}{4\pi^2|\surf|}\int_{\mathbb{R}^2}\big(\widehat{r}(\bk)-\widehat{r}(\bk_n)\big) D(\bk-\bk_n)\,D(\bk-\bk_m)\,\dd\bk ,
\end{equation}
so the error is driven solely by the variation of the weight. Since $|D(\bk-\bk_n)D(\bk-\bk_m)|\le\tfrac12 D(\bk-\bk_n)^2+\tfrac12 D(\bk-\bk_m)^2$, it suffices to bound $\int|\widehat{r}(\bk)-\widehat{r}(\bk_n)|\,\dd\mu_n(\bk)$ for the unit-mass measures $\dd\mu_n\triangleq D(\bk-\bk_n)^2\,\dd\bk/(4\pi^2|\surf|)$.

Step 2 (tail mass of $\mu_n$). From $|\int_{-S/2}^{S/2}e^{\im ks}\dd s|\le\min(S,2/|k|)$ and the one-dimensional Parseval identity $\int_{\mathbb{R}}S^2\sinc^2(kS/2)\,\dd k=2\pi S$, the mass of $\mu_n$ outside the ball $\mathcal{B}=\{\|\bk-\bk_n\|\le\epsilon\}$, where at least one coordinate exceeds $\epsilon/\sqrt2$, satisfies
\begin{equation}\label{eq:tailmass}
\mu_n(\mathcal{B}^{\rm c})\le\frac{1}{4\pi^2|\surf|}\cdot\frac{8\sqrt2}{\epsilon}\,2\pi\,(S_x+S_y)=O\Big(\frac{1}{\epsilon\,S_{\min}}\Big).
\end{equation}

Step 3 (three regions). Fix $\epsilon\in(0,\delta/2)$ and split $\mathbb{R}^2$ into $\mathcal{B}$, the bulk $\{\|\bk\|\le\kappa-\delta/2\}\setminus\mathcal{B}$, and the rim ring $\{\kappa-\delta/2<\|\bk\|<\kappa\}$. Outside the disc the weight vanishes and only the constant $\widehat{r}(\bk_n)$ survives, controlled by \eqref{eq:tailmass}. On $\mathcal{B}$, where the weight is smooth, the contribution is at most the modulus of continuity $\omega(\epsilon)$. On the bulk, the integrand is bounded by $2\sup_{\|\bk\|\le\kappa-\delta/2}\widehat{r}$ and the contribution is $O(1/(\epsilon S_{\min}))$ by \eqref{eq:tailmass}. On the ring, H\"older with exponents $p\in(1,2)$ and $p'=p/(p-1)$ applies: $\widehat{r}\in L^{p}(\text{ring})$ since $(\kappa^2-\varrho^2)^{-p/2}$ is integrable for $p<2$, while the tail computation of Step 2, with the constrained axis at distance at least $\delta/(2\sqrt2)$ from $\bk_n$ and raised to the power $p'$, gives $\|D(\cdot-\bk_n)^2/(4\pi^2|\surf|)\|_{L^{p'}(\text{ring})}=O(S_{\min}^{-1/p'})$.

Step 4 (conclusion). With $\epsilon=S_{\min}^{-1/2}$, $|\varepsilon_{nm}|\le\omega(S_{\min}^{-1/2})+O(S_{\min}^{-1/2})+O(S_{\min}^{-1/p'})\to0$ for any $q=1/p'<1/2$. For evanescent modes, $\|\bk_n\|\ge\kappa+\delta$, the ball lies outside the weight's support, the smooth-region term is absent, and the same bounds give $[\bB_R]_{nn}\to Z_s$. $\hfill\blacksquare$

\section{Proof of Lemma~\ref{lem:grad}}\label{app:grad}
Derivative of $\rho$ (and of $\vartheta$). From $\rho=\tfrac32[f_0(x)+j_2(x)c^2]$, $x=\kappa\tau$, $c=\bar{\bp}^{\mathsf T}\hat{\boldsymbol{\tau}}$, the chain rule needs $\nabla_{\bar{\boldsymbol{\tau}}}x=\kappa\hat{\boldsymbol{\tau}}$ and $\nabla_{\bar{\boldsymbol{\tau}}}c=\nabla(\bar{\bp}^{\mathsf T}\bar{\boldsymbol{\tau}}/\tau)=\bar{\bp}/\tau-(\bar{\bp}^{\mathsf T}\bar{\boldsymbol{\tau}})\bar{\boldsymbol{\tau}}/\tau^3=(\bI_2-\hat{\boldsymbol{\tau}}\hat{\boldsymbol{\tau}}^{\mathsf T})\bar{\bp}/\tau$, whence \eqref{eq:gradrho} with $\nabla(c^2)=2c\nabla c$. The scalar derivatives follow from the standard recurrences: $j_0'=-j_1$ and $j_\ell'=j_{\ell-1}-\tfrac{\ell+1}{x}j_\ell$, so $j_2'=j_1-3j_2/x$ and
\begin{align*}
\Big(\frac{j_1}{x}\Big)'&=\frac{j_1'x-j_1}{x^2}=\frac{(j_0-2j_1/x)x-j_1}{x^2}=\frac{j_0}{x}-\frac{3j_1}{x^2}
\\ &\Rightarrow\
f_0'=-j_1-\frac{j_0}{x}+\frac{3j_1}{x^2}.
\end{align*}
Since $y_\ell$ and $h_\ell=j_\ell+\im y_\ell$ satisfy the same recurrences, $\nabla\vartheta$ is \eqref{eq:gradrho} with $j_\ell\mapsto h_\ell$. Equation \eqref{eq:gradC} then follows from the evenness of $\rho$ and the position-independence of the diagonal of $\bC_\varepsilon$. For \eqref{eq:gradV}, write $v_{a,k}=(1+\varepsilon+\im\zeta_A)w_{a,k}+\sum_{b\ne a}\vartheta(\bu_a-\bu_b)w_{b,k}$. Differentiating $|v_{a',k}|^2$ in $\bu_a$ and using $\nabla_{\bu_a}\vartheta(\bu_{a'}-\bu_a)=-\nabla\vartheta(\bu_{a'}-\bu_a)$ with the evenness of $\vartheta$ collects the two contributions, the port's own voltage through its neighbors and every neighbor's voltage through the port, into the stated pairwise form.

Rate gradient \eqref{eq:gradrate}. Only the $a$-th entry of $\bh_k$ depends on $\bu_a$, so write $\mathbf{d}=\mathbf{d}_{k,a}=\nabla_{\bu_a}[\bh_k]_a$ and differentiate componentwise in the real vector $\bu_a$. For $S_k=|\bh_k^{\mathsf H}\bw_k|^2$: since $\bh_k^{\mathsf H}\bw_k=\sum_b[\bh_k]_b^{*}w_{b,k}$, its gradient is $w_{a,k}\mathbf{d}^{*}$, and for any complex differentiable-in-$\bu$ scalar $z$, $\nabla|z|^2=2\Real\{z^{*}\nabla z\}$. Hence $\nabla S_k=2\Real\{(\bw_k^{\mathsf H}\bh_k)\,w_{a,k}\mathbf{d}^{*}\}$, using $(\bh_k^{\mathsf H}\bw_k)^{*}=\bw_k^{\mathsf H}\bh_k$. Identically, $\nabla I_k=2\sum_{i\ne k}\Real\{(\bw_i^{\mathsf H}\bh_k)\,w_{a,i}\mathbf{d}^{*}\}$. Then, from $\Gamma_k=S_k/I_k$,
\begin{equation*}
\nabla\log_2(1+\Gamma_k)=\frac{\nabla S_k\ I_k-S_k\nabla I_k}{\ln2\ I_k^2\,(1+S_k/I_k)}
=\frac{\nabla S_k-\Gamma_k\nabla I_k}{\ln2(I_k+S_k)},
\end{equation*}
and substituting the two gradients and collecting the common factor $\mathbf{d}^{*}$ inside the real part yields \eqref{eq:gradrate} with the stated $\zeta_{k,a}$.

Channel gradient \eqref{eq:gradh}. Write $[\bh_k]_a=\im\kappa Z_0\ell\,F(\bu_a)$ with $F(\bu)=a(R)\,(\boldsymbol{\psi}_k^{\mathsf H}\bp)+b(R)\,(\boldsymbol{\psi}_k^{\mathsf H}\hat{\boldsymbol{\tau}})(\hat{\boldsymbol{\tau}}^{\mathsf T}\bp)$ and $\boldsymbol{\tau}=\br_k-[\bu^{\mathsf T},0]^{\mathsf T}$. The Jacobian of $\boldsymbol{\tau}$ with respect to $\bu$ is $-\boldsymbol{\Pi}^{\mathsf T}$ (each planar coordinate of $\bu$ subtracts from the corresponding component of $\boldsymbol{\tau}$), so $\nabla_{\bu}R=-\boldsymbol{\Pi}\hat{\boldsymbol{\tau}}$ and, differentiating $\hat{\boldsymbol{\tau}}=\boldsymbol{\tau}/R$,
\begin{align*}
\nabla_{\bu}\big(\boldsymbol{\psi}_k^{\mathsf H}\hat{\boldsymbol{\tau}}\big)&=-\frac{1}{R}\,\boldsymbol{\Pi}\big(\bI_3-\hat{\boldsymbol{\tau}}\hat{\boldsymbol{\tau}}^{\mathsf T}\big)\boldsymbol{\psi}_k^{*},\\
\nabla_{\bu}\big(\hat{\boldsymbol{\tau}}^{\mathsf T}\bp\big)&=-\frac{1}{R}\,\boldsymbol{\Pi}\big(\bI_3-\hat{\boldsymbol{\tau}}\hat{\boldsymbol{\tau}}^{\mathsf T}\big)\bp,
\end{align*}
where the first identity is obtained by writing $\boldsymbol{\psi}_k^{\mathsf H}\hat{\boldsymbol{\tau}}=\sum_i \psi_{i}^{*}\hat{\tau}_i$ and applying the Jacobian of $\hat{\boldsymbol{\tau}}$, namely $-(\bI_3-\hat{\boldsymbol{\tau}}\hat{\boldsymbol{\tau}}^{\mathsf T})\boldsymbol{\Pi}^{\mathsf T}/R$, and the symmetry of the projector. The product rule on $F$ is
\begin{align*}
\nabla_{\bu}F&=\big[a'(R)(\boldsymbol{\psi}_k^{\mathsf H}\bp)+b'(R)(\boldsymbol{\psi}_k^{\mathsf H}\hat{\boldsymbol{\tau}})(\hat{\boldsymbol{\tau}}^{\mathsf T}\bp)\big]\big(-\boldsymbol{\Pi}\hat{\boldsymbol{\tau}}\big)
\\&-\frac{b(R)}{R}\,\boldsymbol{\Pi}\big(\bI_3-\hat{\boldsymbol{\tau}}\hat{\boldsymbol{\tau}}^{\mathsf T}\big)\big[(\hat{\boldsymbol{\tau}}^{\mathsf T}\bp)\,\boldsymbol{\psi}_k^{*}+(\boldsymbol{\psi}_k^{\mathsf H}\hat{\boldsymbol{\tau}})\,\bp\big],
\end{align*}
which is \eqref{eq:gradh} after multiplying by $\im\kappa Z_0\ell$. Finally, differentiating $a$, $b$ of \eqref{eq:greenclosed} with $g'=(\im\kappa-1/R)g$:
\begin{align*}
a'(R)&=g'(R)\Big(1+\tfrac{\im}{\kappa R}-\tfrac{1}{(\kappa R)^2}\Big)+g(R)\Big(-\tfrac{\im}{\kappa R^2}+\tfrac{2}{\kappa^2R^3}\Big),\\
b'(R)&=-g'(R)\Big(1+\tfrac{3\im}{\kappa R}-\tfrac{3}{(\kappa R)^2}\Big)-g(R)\Big(-\tfrac{3\im}{\kappa R^2}+\tfrac{6}{\kappa^2R^3}\Big).
\end{align*}
Every term is elementary in $R$, so the gradient evaluates in closed form. $\hfill\blacksquare$

\bibliographystyle{IEEEtran}
\bibliography{refs}

\end{document}